\documentclass[journal,12pt,onecolumn,draftclsnofoot,]{IEEEtran}

\usepackage{ifpdf}
\usepackage{cite}

\ifCLASSINFOpdf
   \usepackage[pdftex]{graphicx}
\else
\fi
\usepackage{array}
\usepackage{url}
\usepackage{epstopdf}
\usepackage{amsmath}
\usepackage{algorithm}
\usepackage{algpseudocode}
\usepackage{multicol}
\usepackage{amssymb}
\usepackage{amsfonts,graphicx,amsthm, lscape}
\usepackage[font=small]{caption}
\usepackage[font=footnotesize]{subcaption}
\usepackage{color}
\usepackage{mathtools}
\usepackage[subnum]{cases}
\usepackage{empheq}
\usepackage{cleveref}
\usepackage{color,amsmath}
\usepackage{booktabs}
\newcommand{\tabletitleunderrule}{\midrule[\heavyrulewidth]}
\newcommand{\rev}[1]{{\color{black}#1}}
\newcommand{\changen}[1]{{\color{black}#1}}

\newcommand{\seprev}[1]{{\color{black}#1}}
\newtheorem{prop}{Proposition}

\ifCLASSINFOpdf
\else
\fi
\begin{document}
\allowdisplaybreaks
\title{Multiple Access Computational Offloading: \rev{Communication Resource Allocation in the} Two-User Case (Extended Version)}
\author{{Mahsa Salmani and Timothy N. Davidson}
\thanks{A condensed version of this manuscript will be submitted for possible publication.}
\thanks{This work was supported in part by the Natural Sciences and Engineering Research Council
of Canada under grant RGPIN-2015-06631. The authors are with the Department of Electrical and Computer Engineering, McMaster University, Hamilton, Ontario, Canada. Email: \{salmam, davidson\}@mcmaster.ca. Preliminary versions of portions of this manuscript appear in \textit{Proc.\ 17\textsuperscript{th} IEEE Int. Wkshp Signal Process.\ Adv.\ Wireless Commun.}, Jun 2016, and in \textit{Conf.\ Rec.\ 51\textsuperscript{st} Asilomar Conf.  Signals, Syst. Comput.}, Oct 2017, and in \textit{Proc.\ 23\textsuperscript{rd} Asia-Pacific Conf.\ Commun.} Dec 2017.  }}

\maketitle

\begin{abstract}
By offering shared computational facilities to which mobile devices can offload their computational tasks, the mobile edge computing framework is expanding the scope of applications that can be provided on resource-constrained devices. When multiple devices seek to use such a facility simultaneously, both the available computational resources and the available communication resources need to be appropriately allocated. In this manuscript, we seek insight into the impact of the choice of the multiple access scheme by developing solutions to the mobile energy minimization problem in the two-user case with plentiful shared computational resources. In that setting, the allocation of communication resources is constrained by the latency constraints of the applications, the computational capabilities and the transmission power constraints of the devices, and the achievable rate region of the chosen multiple access scheme. For both indivisible tasks and the limiting case of tasks that can be infinitesimally partitioned, we provide a closed-form and quasi-closed-form solution, respectively, for systems that can exploit the full capabilities of the multiple access channel, and for systems based on time-division multiple access (TDMA). For indivisible tasks, we also provide quasi-closed-form solutions for systems that employ sequential decoding without time sharing or independent decoding. Analyses of our results show that when the channel gains are equal and the transmission power budgets are larger than a threshold, TDMA (and the suboptimal multiple access schemes that we have considered) can achieve an optimal solution. However, when the channel gains of each user are significantly different and the latency constraints are tight, systems that take advantage of the full capabilities of the multiple access channel can substantially reduce the energy required to offload.
\end{abstract}

\IEEEpeerreviewmaketitle
\section{Introduction}
\label{Intro}
\IEEEPARstart{T}{he} widespread adoption of mobile computing devices and wireless communication networks has enabled the development of applications and services that previously could only be envisioned; e.g., \cite{weiser1991computer}. The success of these developments is fuelling the ambition for future applications and services, but as that ambition has grown, the modest computational, storage and energy resources of the mobile devices have become significant constraints. The mobile cloud, mobile edge, and fog computing frameworks seek to address those constraints by providing shared computational resources to which mobile devices can offload their computational tasks, or a portion thereof; e.g.,\cite{satyanarayanan2009case, kumar2010cloud, liu2013gearing, miettinen2010energy, kumar2013survey, fernando2013mobile, barbera2013offload, mao2017survey}. Offloading offers the potential for the mobile device to obtain the results of computationally-intensive or memory-intensive tasks more quickly than would be possible using local computation, and it also offers the potential to better manage the battery life of the device. Some early prototypes of mobile computational offloading systems include MAUI \cite{cuervo2010maui}, CloneCloud \cite{chun2011clonecloud}, and ThinkAir \cite{kosta2012thinkair}.

The offloading opportunities provided by the mobile cloud computing framework need to be balanced against the energy required to communicate \rev{reliably} with the networking infrastructure that connects to the shared computing resources \cite{barbera2013offload}, the latency of that communication, the contention for the limited communication resources of the network \cite{lei2013challenges}, and the contention for the limited computational resources of the cloud \cite{liu2013gearing}. Indeed, the problem of deciding when and how to exploit the resources provided by the mobile cloud computing framework can be formulated as a joint optimization problem over the available computational and communication resources; e.g., \rev{ \cite{sardellitti2015joint, chen2015semidefinite, munoz2015optimization, munoz2014energy, salmani2016multiple, wang2017optimized, wang2017joint}}. That formulation typically captures the energy that would be required to complete the computational task locally (on the mobile device) and the latency incurred in doing so, and the energies and latencies associated with transmitting the required information to the shared computational resources, completing the task there and returning the results to the mobile device. When the components of the task at hand are tightly coupled, the task is often considered to be indivisible, and hence the decision of whether or not to offload the task is a binary decision. When the task can be partitioned into separate components, the system can take advantage of the implicit parallelism between the mobile user and the access point. As a result, the formulation of the offloading problem may include decisions on which components to offload, or, in the limit, what fraction of the task to offload. When there are multiple devices that are seeking access to the computing resources, the architecture of the envisioned system will determine whether the offloading decisions are to be made centrally, or in a distributed fashion. 
As this discussion suggests, in the general case, the problem of deciding whether or not to offload (a fraction of) a computational task can be a computationally demanding problem in and of itself. Therefore, the development of insight into the structure of good solutions has the potential to guide the development of practical algorithms.

In this manuscript we seek to develop insight into the impact of the choice of the multiple access scheme in a multiuser offloading system in which the allocation of resources is performed centrally. In previous work on such systems, the multiple access scheme has been chosen \textit{a priori}. For example, the users' channels may be assumed to be orthogonal (as they are in time division multiple access, TDMA, and frequency division multiple access, FDMA) \cite{sardellitti2015joint, chen2016joint, chen2016multi, you2017energy}, or it may be assumed that the access point performs independent decoding \cite{sardellitti2015joint} or ordered minimum mean square error successive interference cancellation \cite{wang2017optimized}. These choices can limit the set of rates at which the users can communicate reliably with the access point, and hence they can constrain the potential of computational offloading. In the main contribution of this manuscript, we do not place constraints on the multiple access scheme and that enables us to take the advantage of the full capabilities of the wireless channel. 

In order to develop insight into the impact of the choice of the multiple access scheme and the corresponding allocation of communication resources, we will consider a two-user scenario with a single access point that is equipped with plentiful computational resources. Each user's task has a separate latency constraint, and we will consider both indivisible tasks and the limiting case of infinitesimally divisible tasks. The goal is to make the offloading decisions (or choose the offloaded fractions) and to allocate the communication resources so that the energy expended by the users is minimized. The decisions and allocations are made centrally, using the knowledge of the (single-input single-output) channels from all users. For indivisible and infinitesimally divisible computational tasks we will derive closed-form and quasi-closed-form solutions, respectively, to the energy minimization problem when the full capabilities of the multiple access channel are exploited. We will also provide corresponding solutions for some simplified multiple access schemes, namely, time-division multiple access (TDMA), and, in the case of indivisible tasks, sequential decoding without time sharing (SDwts), and independent decoding (ID). 

\subsection{Principles of Proposed Approach}
The broad principles that underlie our approach to the resource allocation problem arise from the observation that the communication rates employed by each user must lie in the intersection of two regions. First, the rates must be large enough to meet the latency constraints imposed by the computational tasks. Second, the rates must lie within the achievable rate region for the chosen multiple access \rev{scheme}. For the simple two-user case that we will consider in this manuscript (see Secs~\ref{Specific_contributions} and \ref{sys_model}), these regions are illustrated in Fig.~\ref{feas_region_all}, where the shaded regions are the achievable rate region for (a) independent decoding, (b) TDMA, and (c) sequential decoding without time sharing, and (d) the capacity region of the multiple access channel,\footnote{Rate pairs on the ``dominant face'' of the capacity region can be achieved by joint decoding or by employing ``time sharing'' between the ``corner points'' of the region, each of which can be achieved by sequential decoding \cite{cover2012elements}.} for a given channel environment and a given set of operating power constraints. The L-shaped dashed lines in each figure denote the boundary of the set of rate pairs that will enable the latency constraints of the given applications to be satisfied. Rate pairs above and to the right of the boundary will enable the latency constraints to be satisfied, and we will call the set of those rate pairs the latency region. In the scenario marked $L$, the latency constraints are quite long and in all four cases there is an intersection between the achievable rate region for the channel and the latency region. Therefore, for each multiple access scheme there are rate pairs that are feasible for the computational offloading problem. In the scenario marked $S$, the latency constraints are quite short and there is no intersection in the cases of independent decoding or TDMA. However, if one takes advantage of the full capabilities of the multiple access channel there is an intersection (see Fig.~\ref{feas_region_all_MAC}), and there are pairs of rates that will enable the latency constraints to be satisfied.

Another useful interpretation of the components in Fig.~\ref{feas_region_all} arises from the fact that the operating power levels of the transmitters control the size of the achievable rate regions: as the operating power is reduced, those regions shrink. If we consider the latency region for long latency constraints, marked by $L$, it appears that there is not much ``room" to shrink the independent decoding and TDMA regions while retaining an intersection with the latency region. However, there appears to be considerable room to shrink the capacity region. This suggests that by taking advantage of the full capabilities of the multiple access channel we can reduce the energy required to offload the latency-constrained tasks.

\subsection{Specific Contributions of the Manuscript}
\label{Specific_contributions}
The specific problems considered in this manuscript concern a two-user system in which each user has a latency-constrained task that they wish to complete with the possible assistance of a single-antenna access point with plentiful computationally resources. We consider both indivisible computational tasks, for which a binary offloading decision must be made, and infinitesimally divisible tasks, for which the fraction of the task to be offloaded is to be determined. In terms of communication resources, each user has a single antenna and a maximum allowable transmission power. The communication channels are assumed to be known and constant over the latency interval. A distinguishing feature of our system is the observation that when users are seeking to offload (a fraction of) their tasks, the communication system can operate in a combination of three modes: one in which both users are transmitting and the others in which only one user is transmitting. In addition to making the offloading decisions, or determining the fractions to be offloaded, our goal includes determining the durations of the time slots in which the system operates in each mode and the rates and powers allocated to each slot, so that the sum of the computational and communication energies expended by the mobile devices is minimized. 

\subsubsection{Indivisible Tasks---Binary Computational Offloading}
One of our key results in the case of indivisible tasks is a closed-form solution to the mobile energy minimization problem in the case in which the full capabilities of the multiple access channel are exploited when both users are offloading. That solution shows that only two of the three available time slots are required. We then obtain a closed-form expression for the optimal solution in the TDMA case. Quite naturally, that solution can also be achieved in two time slots. In the cases of independent decoding and sequential decoding without time sharing, the optimal solution may have three active time slots, and we provide quasi-closed-form solutions for both those cases. These solutions each depend on the solution of different three-variable optimization problems. Based on the structure of each of those problems we propose a coordinate descent method in which each subproblem is convex. That approach is guaranteed to converge to a stationary point \cite[Theorem 1]{hong2016unified}, and in all our numerical experiments it converged to the globally optimal solution. 

Our closed-form solutions enable us to show that when TDMA is able to satisfy the latency constraints, the optimized solution in the case of independent decoding takes the form of TDMA. (In TDMA systems the decoders operate independently.)
We will also show that when the channel gains are the same, if the power budgets of the users are above an explicit threshold, the optimized TDMA solution is globally optimal. A consequence of that result is that independent decoding and sequential decoding (wts) are also optimal when the channel gains are the same.

In a complementary way, our numerical results will illustrate that when the channel gains are quite different and the latency constraints are reasonably tight, taking advantage of the full capabilities of the multiple access channel offers significant reductions in the energy required to offload the applications and substantially broadens the range of channel gains over which offloading can be achieved while respecting the power constraints of the devices. Our results will also show that a large fraction of these gains can be achieved by sequential decoding without time sharing. 
\begin{figure}
\begin{subfigure}{.5\linewidth}
\centering
\includegraphics[scale=.5]{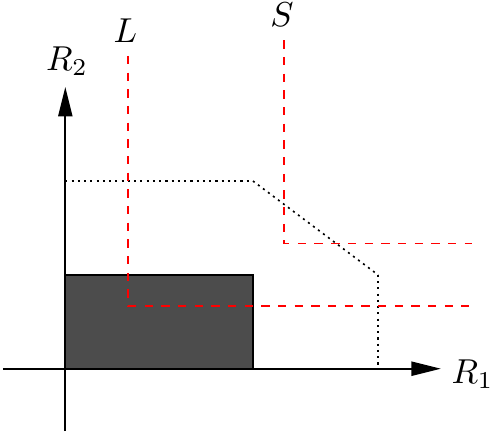}
\caption{\centering Independent decoding}
\label{feas_region_all_ind}
\end{subfigure}%
\begin{subfigure}{.5\linewidth}
\centering
\includegraphics[scale=.5]{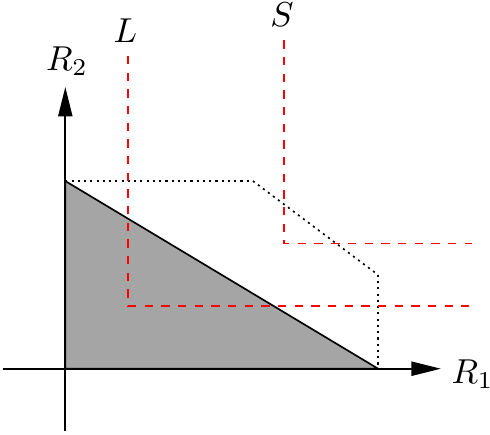}
\caption{\centering TDMA}
\label{feas_region_all_TDMA}
\end{subfigure}\\%
\begin{subfigure}{.5 \linewidth}
\centering
\includegraphics[scale=.5]{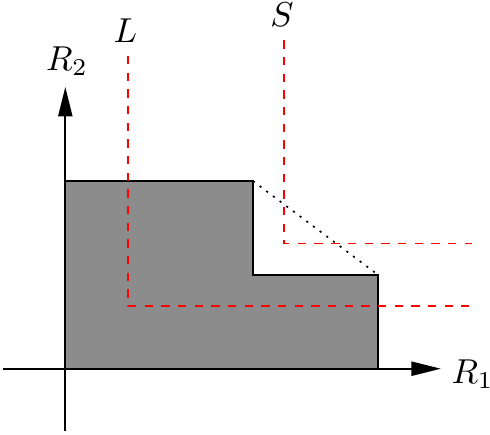}
\caption{\centering Sequential decoding  without time sharing}
\label{feas_region_all_sec_dec}
\end{subfigure}%
\begin{subfigure}{.5 \linewidth}
\centering
\includegraphics[scale=.5]{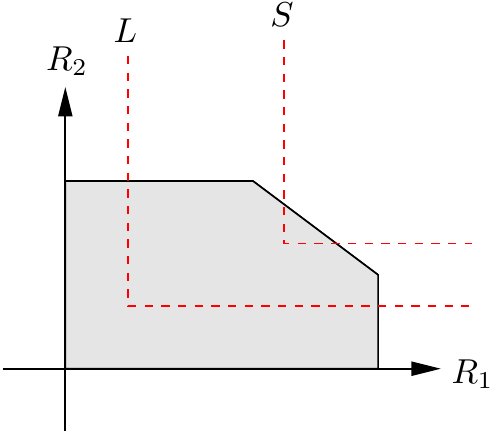}
\caption{\centering Full capabilities of the  multiple access channel}
\label{feas_region_all_MAC}
\end{subfigure}%
\caption{The set of feasible rates for \changen{a simple two-user system that employs a single time slot} is the intersection of the achievable rate region (shaded) and the region of the rates that will satisfy the latency constraints, which is above and to the right of the L-shaped dashed lines.}
\label{feas_region_all}
\end{figure}

\subsubsection{Infinitesimally Divisible Tasks---Partial Computational Offloading}
In Section~\ref{partial_offloading} we extend our approach to the limiting case of infinitesimally divisible computational tasks that may be offloaded using the full multiple access scheme or TDMA. Our focus in that section is on data-partitionable applications \cite{wang2016mobile}, in which a simple-to-describe operation is applied, independently, to different blocks of data.
For both the full multiple access and TDMA schemes, we obtain quasi-closed-form solutions that depend on the solution of (different) three-variable optimization problems. For the full multiple access scheme a coordinate descent approach that is guaranteed to yield a stationary point  is employed (see \cite[Theorem 1]{hong2016unified}), and in the TDMA case the problem is convex and hence, an optimal solution can be easily obtained. The structure of these subproblems enables us to show that, as in the case of indivisible tasks, when the channel gains are the same, if the transmission power budgets are above a threshold, then the optimized TDMA partial offloading solution is globally optimal. However, when the channel gains are quite different, performing complete offloading using full multiple access scheme can result in significantly lower mobile energy consumption than partial offloading using TDMA. 
\changen{At the end of Sec.~\ref{partial_offloading} we briefly outline the extension of our approach to the ``mixed'' case, in which one user has an indivisible task and the other has a divisible task.}

\section{System Model} 
\label{sys_model}
The goal of this manuscript is to develop insight into the impact of the choice of the multiple access scheme on the energy consumed by a computational offloading system. To do so, we will consider a two-user system in which each user seeks to obtain the results of a latency-constrained computational task with the possible assistance of a single access point with plentiful computational resources. The nature of the computational tasks that the users are to execute has a significant impact on the way this problem is formulated; e.g.,\cite{munoz2015optimization, wang2016mobile}. If the components of the task are tightly coupled, the problem must be executed either by the user or by the access point alone. That is, the task must be completely offloaded or not offloaded at all; e.g.,\cite{kumar2010cloud, wu2013tradeoff, sardellitti2015joint}. In contrast, tasks with independent or loosely coupled components can benefit from the parallelism between the mobile device and the access point, with a portion of the task being offloaded and the remainder being computed locally; e.g., \cite{zhang2013energy, wang2016mobile}. In this manuscript, we will first consider the case in which both users have an indivisible task; see Sections~\ref{binary_offloading}--\ref{insight}. Then, in Section~\ref{partial_offloading}, we consider the limiting case in which both users have an infinitesimally divisible task that can be partially offloaded, with the remainder of the task computed locally. For ease of exposition, in this section we will establish the system model for the case of indivisible tasks. The extension of this model to infinitesimally divisible tasks is provided in Section~\ref{partial_offloading}.

In the case of indivisible tasks, the access point will decide whether or not each user will offload its task. It is assumed that the access point knows the energy that each user would expend in order to complete its task locally before the latency deadline. (If the task cannot be completed locally in time that energy is notionally set to $+\infty$.) That local computational energy is then compared to the energy that would be expended to offload the task to the access point in such a way that the result can be returned to the user before the deadline. That transmission energy is dependent on the allocation of the available communication resources. 

In the general case, the users' tasks will have different description lengths and different latency constraints, and the users will offload their descriptions at different rates. As a result, when both users are offloading, one user may complete its transmission before the other. In order to take advantage of that fact in minimizing the users' energy consumption, we will adopt the time slotted structure in Fig.~\ref{time_slots}. In the first time slot, both users are using the channel, and in the second and third time slots, user 1, the user with smaller latency, and user 2 complete the offloading of their applications, respectively. Since the length of any time slot can be set to zero, this time slotted structure naturally incorporates scenarios in which only one user is offloading. 
\begin{figure}
\centering 
\includegraphics[width=0.3\linewidth]{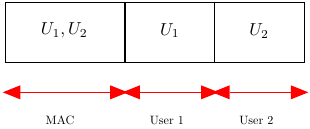} 
\caption{Different time slots in computational offloading for two users. In the first time slot both users are offloading simultaneously, while in the second only user one is offloading and in the third only user two is offloading. } 
\label{time_slots} 
\end{figure} 

The problem that we will consider is to minimize the total energy consumption of the users. For each user the energy consumption is either the local computation energy or the offloading energy. The variables that we can manipulate are the offloading decisions, the duration of each time slot (or equivalently, the fractions of the number of bits offloaded in each time slot), and each user's transmission power and rate in each slot. \rev{The constraints on those variables arise from} the latency of each offloaded task, the bounds on the transmission power of each user, and the \rev{set of achievable} rates of the chosen multiple access scheme. That is, we seek solutions to problems of the form
\begin{subequations}  
\label{prob_formula_main_inword}
\begin{align}
\min_{\substack{\text{rates},\hspace{5pt} \text{powers}, \\ \text{time slot durations, } \\ \text{offloading decisions}} } \quad   & \text{Transmission energy + Local computational energy}\\ 
\label{prob_formula_main_inword_b}
 \text{s.t.} \quad \quad \quad  & \text{Latency constraint of each task}, \\
\label{prob_formula_main_inword_c}
& \text{Power constraint on each user},\\
\label{prob_formula_main_inword_d}
& \text{Achievable rate region}. 
\end{align}
\end{subequations}

In order to formulate the problem described in \eqref{prob_formula_main_inword}, let $B_k$ denote the total number of bits needed to describe the problem of user $k$ and let $\gamma_{ki}$ denote the fraction of bits that user $k$ transmits at time slot $i$. The number of bits transmitted by user $k$ in time slot $i$ can then be written as $\gamma_{ki} B_k$. (We will assume that $B_k$ is large enough that $\gamma_{ki}$ can be modeled as a continuous variable in $[0, 1]$.) We will let $T_s$ denote the symbol interval of the system, and we let $\tau_i$ denote the length (in channel uses) of the $i$\textsuperscript{th} time slot. Hence, the $i$\textsuperscript{th} time slot has a duration $\tau_i T_s$. If $P_{ki}$ and $R_{ki}$ denote the transmission power and data rate (in units per channel use) for an offloading user, $k$, in time slot $i$, then the offloading energy consumption of user $k$ is $E_{\text{off}_k} = \sum_i P_{ki} \tau_{i}$. 
We will let $\bar{P}_k$ denote the maximum operating power level of user $k$ and we impose the power constraint $P_{ki} \le \bar{P}_k$.
The energy consumption of local execution, $E_{\text{loc}_k}$, and the time that would be needed to complete the task locally, $t_{\text{loc}_k}$, are determined by the number of CPU cycles required to execute the computation task and the structure of user $k$'s CPU. In the binary computation offloading case, these parameters are constant and are assumed to be known by the access point. 

We will adopt a computational model in which the full description of an offloaded task must be received by the access point before computation can begin, and the results only become available once the entire task has been completed. Under that model, the latency of an offloaded task is the sum of the time taken to transmit the description of the problem to the access point, $t_{\text{UL}_k}$, the time taken to compute the result at the access point, $t_{\text{exe}_k}$, and the time taken to return the result to the user, $t_{{\text{DL}}_k}$. Accordingly, the latency constraint for user $k$ can be written as 
\begin{equation}
\label{latency_const}
t_{\text{UL}_k}+t_{{\text{exe}}_k}+t_{\text{DL}_k} \le L_k.
\end{equation}
As outlined in the Introduction, we will focus on scenarios in which there are sufficient computational resources available at the access point so that the execution time of the offloaded task of the $k^{th}$ user can be considered independent of the computational loads imposed by the other users. Furthermore, we will assume that sufficient communication resources are available on the downlink for the time taken to return the result to user $k$ to be considered as a constant. 
In other words, we assume that it is the allocation of communication resources on the uplink that is the bottleneck of the offloading problem. 

The uplink communication environment that we will consider is a narrowband multiple access channel with two single-antenna users offloading to a single-antenna access point. If the signal transmitted by the $k$\textsuperscript{th} user at a given time instant is denoted by $s_k$ and the corresponding channel is denoted by $h_k$, then the signal received at the access point is 
\begin{equation}
\label{received_sig}
y= h_1 s_1+ h_2 s_2 +v,
\end{equation}
where $v$ is a sample from a zero mean circular white Gaussian noise process with variance $\sigma^2$. We will let $\alpha_k = \tfrac{|h_k|^2}{\sigma^2}$ denote the effective power gain of the channel of user $k$ and we will assume that $\alpha_1$ and $\alpha_2$ are known by the access point.
To establish the constraints on the rates we will make the assumption that each time slot $\tau_i$ is long enough that the achievable rate region for a finite block length can be approximated by the limiting region for long block lengths. (Recent work suggests that conventional rate limits provide insightful guidelines for communication over finite block length even when the blocks are quite short \cite{polyanskiy2010channel, jazi2012simpler}.)

For ease of  later reference, we have summarized the definitions of the key parameters and variables in Table~\ref{tab:glossary}.
\begin{table}
\centering
\caption{Parameters and variables}
\label{tab:glossary}
\begin{footnotesize}
\begin{tabular}{ll}
\toprule
Symbol & Quantity\\
\tabletitleunderrule
$\alpha_k$ & effective power gain of user $k$, $|h_k|^2/\sigma^2$, \\ \midrule
$\gamma_{ki}$ & the fraction of bits transmitted by user $k$ in slot $i$ \\ \midrule
$L_k$ & latency of user $k$ \\ \midrule
$\tilde{L}_k$ & \seprev{$(L_k-t_{\text{exe}_k}-t_{\text{DL}_k})/T_s$} \\ \midrule
\seprev{$\bar{L}_k$} & \seprev{$L_k-t_{\text{DL}_k}$} \\ \midrule
$P_{ki}$ & power transmitted by user $k$ in slot $i$ \\ \midrule
$\bar{P}_k$ & power constraint on user $k$ \\ \midrule
$R_{ki}$ & rate transmitted by user $k$ in slot $i$ \\ \midrule
$\tau_1$ & \changen{length} of first time slot\\ \midrule
$T_s$ & symbol interval\\  \midrule
$E_{\text{loc}_k}$ & local energy consumption of user $k$\\ 
\bottomrule
\end{tabular}
\end{footnotesize}
\end{table}
\section{Binary Computation Offloading}
\label{binary_offloading}
In this section, we consider the case in which both users have indivisible computational tasks. As outlined in the previous section, in this setting, the access point will evaluate the total mobile energy consumption in each of the four possible cases of local or offloaded computation. The energy of local computation is presumed to be known by the access point (if local computation can meet the latency deadline). In the following section we will find the closed-form optimal solution of the problem in \eqref{prob_formula_main_inword} when only one user is offloading. The formulation of the ``complete'' computation offloading problem, in which it is decided that both users should offload their tasks, will be presented in Section~\ref{complete_offlaoding} and solutions will be obtained in Sections~\ref{general_case} and \ref{suboptimal_schemes}.

\subsection{Single Offloading User}
When only one user is offloading, the duration of the first time slot is zero, (i.e., $\tau_1 = 0$), and the solution to the minimal transmission energy problem has a simple closed-form expression \cite{salmani2016multiple}. If we consider the case in which user~1 is offloading, then the problem in \eqref{prob_formula_main_inword} can be written as
 \begin{subequations}
\label{prob_slots_main_single_user}
\begin{align} \notag
\min_{P_{1}, R_{1}} \quad  &  \displaystyle{(\tfrac{B_1}{R_{1}}) P_{1}}+  E_{\text{loc}_2} \\ 
\label{prob_slots_main_single_user_b} 
\text{s.t.}  \quad   & \tfrac{B_1}{R_{1}} \le \tilde{L}_1,\\
\label{prob_slots_main_single_user_c}
& 0 \le P_{1} \le \bar{P}_1,\\
\label{prob_slots_main_single_user_d}
\quad & 0 \le R_{1} \le \log_2 (1+ \alpha_1 P_{1}),
\end{align}
\end{subequations}
where \seprev{$\tilde{L}_k =\tfrac{L_k- t_{{\text{exe}}_k}-t_{\text{DL}_k}}{T_s} $}. In this setting, the local execution energy consumption of user 2 is constant and it can be shown that if $ \tfrac{B_1}{\tilde{L}_1} \le \log_2(1+\alpha_1\bar{P}_1)$, then  the optimal communication resource allocation to user 1 is $R_1= \tfrac{B_1}{\tilde{L}_1}$, and $P_1 = \tfrac{2^{R_1}-1}{\alpha_1}$. Otherwise user 1 cannot meet its latency constraint by offloading.

\subsection{Both Users Offloading: Complete Computation Offloading}
\label{complete_offlaoding}
In the case in which both users are offloading, we observe that the durations of the second and third time slots can be written as $\tau_2 = \tfrac{B_1- \gamma_{11}B_1}{R_{12}}$ and $\tau_3=\tfrac{B_2-\gamma_{21}B_2}{R_{23}}$, respectively. Using our ordering of the users (so that $L_1 \le L_2$), the problem of minimizing the sum of the users' transmission energies required to meet the latency constraints, subject to the power constraints and the achievable rate region, $\mathcal{R}$, of the chosen multiple access scheme can be formulated as
 \begin{subequations}
\label{prob_slots_main}
\begin{align} 
\min_{\substack {P_{11}, P_{12}, P_{21}, P_{23}, \\ R_{11}, R_{12}, R_{21}, R_{23}, \\ \gamma_{11}, \gamma_{21}, \tau_1}} \quad &  \displaystyle{ \tau_1( P_{11} + P_{21})} 
+  \displaystyle{ \bigl(\tfrac{B_1-\gamma_{11}B_1}{R_{12}}\bigr) P_{12}}+ \displaystyle{ \bigl(\tfrac{B_2-\gamma_{21}B_2}{R_{23}}\bigr) P_{23}} \\ 
\label{prob_slots_main_b1} 
\text{s.t.}  \enspace \quad \quad \quad  & 0 \le \tau_1,\\
\label{prob_slots_main_b} 
 \quad  &  0 \le \gamma_{11}, \gamma_{21} \le 1,\\
 \label{prob_slots_main_b_new} 
 \quad  &  \tau_1 R_{k1} = \gamma_{k1}B_k, \quad k=1, 2,\\
\label{prob_slots_main_c}
\quad & \tau_1+\tfrac{B_1- \gamma_{11}B_1}{R_{12}} \le \tilde{L}_1,\\
\label{prob_slots_main_d}
\quad  & \tau_1+\tfrac{B_1-\gamma_{11}B_1}{R_{12}} +\tfrac{B_2-\gamma_{21}B_2}{R_{23}} \le \tilde{L}_2,\\
\label{prob_slots_main_e}
& 0 \le P_{k1}, P_{k2} \le \bar{P}_k, \quad k=1,2,\\
\label{prob_slots_main_f}
\quad & 0 \le R_{12} \le \log_2 (1+ \alpha_1 P_{12}),\\
\label{prob_slots_main_g}
\quad & 0 \le R_{23} \le \log_2 (1+ \alpha_2 P_{23}),\\
\label{prob_slots_main_h}
\quad & \{R_{11}, R_{21}\} \in \mathcal{R},
\end{align}
\end{subequations}
\changen{where we have used the equality constraints in \eqref{prob_slots_main_b_new} to enhance the connection to the partial offloading in Section~\ref{partial_offloading}.}

One of the key results in this manuscript is a closed-form expression for the optimal solution to the problem in \eqref{prob_slots_main} when the full multiple access scheme is employed; i.e., when $ \mathcal{R}$ in \eqref{prob_slots_main_h} is the capacity region of the multiple access channel; see Section~\ref{general_case}. We also provide a closed-form solution in the case of TDMA (Section~\ref{TDMA}) and quasi-closed-form solutions for the cases of sequential decoding without time sharing (Section~\ref{Sec_Dec}) and independent decoding (Section~\ref{Ind_Dec}). In the case of full multiple access scheme we will show (see Appendix~\ref{Two_Time_Slots}) that only two time slots are required, and, quite naturally, this is also the case for TDMA. However, for the two other multiple access schemes there are scenarios in which all three time slots are employed. 

\section{Complete Computation Offloading: Full Multiple Access Scheme}
\label{general_case}
\seprev{When the full capabilities of the multiple access channel are used in the first time slot, the achievable rate region $\mathcal{R}$ in \eqref{prob_slots_main_h} becomes the capacity region. Using the standard description of that region \cite{cover2012elements}, and using the equality constraints in \eqref{prob_slots_main_b_new} to determine $\gamma_{k1}$, the problem in \eqref{prob_slots_main} can be written as  }
 \begin{subequations}
\label{prob_slots_general_K}
\begin{align}
\label{slots_general_K_a}
\min_{\substack {P_{11}, P_{12}, P_{21}, P_{23}, \\ R_{11}, R_{12}, R_{21}, R_{23}, \tau_1}} &  \tau_1 (P_{11} + P_{21})+  \bigl(\tfrac{B_1-\tau_1R_{11}}{R_{12}}\bigr) P_{12} + \bigl(\tfrac{B_2-\tau_1R_{21}}{R_{23}}\bigr) P_{23} \\ 
\label{slots_general_K_b}
\text{s.t.}   \quad \quad \quad  & \eqref{prob_slots_main_b1}, \eqref{prob_slots_main_e}-\eqref{prob_slots_main_g},\\
\label{slots_general_K_b1}
\quad & 0 \le \tau_1 R_{k1} \le B_k, \quad k=1, 2,\\
\label{slots_general_K_b2}
\quad & \tau_1+\tfrac{B_1- \tau_1 R_{11} }{R_{12}} \le \tilde{L}_1,\\
\label{slots_general_K_b3}
\quad  & \tau_1+\tfrac{B_1- \tau_1 R_{11} }{R_{12}} +\tfrac{B_2- \tau_1 R_{21} }{R_{23}} \le \tilde{L}_2,\\
\label{slots_general_K_c}
\quad  & 0 \le R_{k1} \le \log_2 (1+ \alpha_k P_{k1}), \quad k=1, 2,\\
\label{slots_general_K_e}
\quad & R_{11}+R_{21}  \le \log_2 (1+ \alpha_1 P_{11} + \alpha_2 P_{21} ). \quad \quad \quad \quad \quad \quad    
\end{align}
\end{subequations}

Although the problem in \eqref{prob_slots_general_K} is cast in the generic three time slot setting, we show in Appendix~\ref{Two_Time_Slots} that it is sufficient to consider a two slot system consisting of a multiple access time slot of length $ \tau_1 = \tilde{L}_1$ and a slot in which user 2 transmits alone. In other words, there is an optimal solution to \eqref{prob_slots_general_K} in which $\tau_1 = \tilde{L}_1$, $P_{12} = 0$, \rev{ $R_{12} = 0$,} and $\tau_1R_{11} = B_1$. This result not only simplifies the derivation of an optimal solution to \eqref{prob_slots_general_K}, it also simplifies the implementation of the system. With this simplification, the problem in \eqref{prob_slots_general_K} reduces to  
\begin{subequations}
\label{Two_slots_general}
\begin{align}
\label{Two_slots_genera_a}
\min_{\substack { P_{11}, P_{21}, P_{23}, \\ R_{21}, R_{23}}} \quad &  \tilde{L}_1 (P_{11} + P_{21}) + \bigl(\tfrac{B_2- \tilde{L}_1R_{21}}{R_{23}}\bigr) P_{23} \\ 
\label{Two_slots_genera_b}
\quad & \tfrac{B_1}{L_1} \le \log_2 (1+ \alpha_1 P_{11}), \\
\label{Two_slots_genera_c}
\quad & 0 \le R_{21} \le \log_2 (1+ \alpha_2 P_{21}), \\
\label{Two_slots_genera_d}
\quad & \tfrac{B_1}{L_1}+R_{21}  \le \log_2 (1+ \alpha_1 P_{11} + \alpha_2 P_{21} ), \\ 
\label{Two_slots_genera_e}
\quad & 0 \le R_{23} \le \log_2 (1+ \alpha_2 P_{23}), \\
\label{Two_slots_genera_f}
\quad  &  \tilde{L}_1+\tfrac{B_2-\tilde{L}_1R_{21}}{R_{23}} \le \tilde{L}_2,\\
\label{Two_slots_genera_g}
& 0 \le P_{11} \le \bar{P}_1, \quad 0 \le P_{21}, P_{23} \le \bar{P}_2.  
\end{align}
\end{subequations}

Our approach to solving the problem in \eqref{Two_slots_general}, and indeed the other problems that we will consider in this manuscript, will be to (i) (precisely) decompose the problem into inner and outer problems, (ii) determine a closed-form or quasi-closed-form expression for the optimal solution of the inner problem in terms of the variables of the outer problem, and (iii) solve the outer problem and subsequently obtain optimal values for the inner problem, e.g., \cite{salmani2016multiple}.
In the first step, we decompose the problem in \eqref{Two_slots_general} in such a way that the transmission rate and transmission power of the time slot in which user 2 transmits alone can be obtained in terms of the rates and powers of the first time slot, namely, 
\begin{eqnarray}
\label{decomposed_prob_general_case}
&\displaystyle{\min_{P_{11}, P_{21},R_{21}}}  \hspace{90pt} &\displaystyle{\min_{P_{23}, R_{23}}} \hspace{8pt}    \eqref{Two_slots_genera_a}\\ 
&\text{s.t.}   \hspace{8pt}   \eqref{Two_slots_genera_b}-\eqref{Two_slots_genera_d}, \eqref{Two_slots_genera_g}  &\hspace{9pt} \text{s.t.} \hspace{13pt}   \eqref{Two_slots_genera_e}-\eqref{Two_slots_genera_g}. \nonumber
\end{eqnarray}
Since the first part of the objective function in \eqref{Two_slots_general} is independent of $P_{23}$ and $R_{23}$, the inner optimization in \eqref{decomposed_prob_general_case} is
\begin{subequations}
\label{general_inner_1}
\begin{align} 
\label{general_inner_1_a}
\min_{P_{23}, R_{23}} \quad & \bigl(\tfrac{B_2-\tilde{L}_1 R_{21}}{R_{23}}\bigr) P_{23} \\ 
\label{general_inner_1_b}
\quad & 0 \le R_{23} \le \log_2 (1+ \alpha_2 P_{23}),\\
\label{general_inner_1_c}
\quad & 0 \le P_{23} \le \bar{P}_2,\\
\label{general_inner_1_d}
\quad  & \tilde{L}_1+\tfrac{B_2-\tilde{L}_1R_{21}}{R_{23}} \le \tilde{L}_2.
\end{align}
\end{subequations}
For a given value of $R_{23}$ the objective in \eqref{general_inner_1} is increasing in terms of $P_{23}$. Since the lower bounds on $P_{23}$ are separable, its optimal value is the minimum feasible value; \rev{i.e., 
\begin{equation}
\label{Opt_P_23}
P_{23} = \tfrac{2^{R_{23}}-1}{\alpha_2}.
\end{equation}}
That enables us to reduce the inner optimization problem to
\begin{subequations}
\label{general_inner_1_R}
\begin{align} 
\label{general_inner_1_R_a}
\min_{R_{23}} \quad & \bigl(\tfrac{B_2-\tilde{L}_1 R_{21}}{\alpha_2}\bigr) \bigl(\tfrac{2^{R_{23}}-1}{R_{23}}\bigr)\\ 
\label{general_inner_1_R_b}
\quad &  0\le R_{23} \le \log_2 (1+ \alpha_2 \bar{P}_{2}), \\
\label{general_inner_1_R_c}
\quad  & \tilde{L}_1+\tfrac{B_2-\tilde{L}_1 R_{21}}{R_{23}} \le \tilde{L}_2.
\end{align}
\end{subequations}
The constraint in  \eqref{general_inner_1_R_b} is obtained from the constraint in \eqref{general_inner_1_c} and it guarantees the feasibility of the problem in \eqref{general_inner_1} in terms of $P_{23}$. 
As shown in \rev{Appendix~\ref{increaseing_pbj_func}}, the objective function in \eqref{general_inner_1_R} is increasing in terms of $R_{23}$. \rev{Since \eqref{general_inner_1_R_c} imposes a lower bound on $R_{23}$ and the right hand side of \eqref{general_inner_1_R_b} imposes an upper bound, the optimal value of} $R_{23}$ is obtained when equality holds in \eqref{general_inner_1_R_c}, so long as that value satisfies \eqref{general_inner_1_R_b}. If it does not, the problem in \eqref{general_inner_1_R} is infeasible and so is that in \eqref{prob_slots_general_K}. Therefore, \rev{when it is feasible} the optimal value of $R_{23}$ is 
\begin{equation}
\label{Opt_R_23}
R_{23} = \tfrac{B_2-\tilde{L}_1 R_{21}}{\tilde{L}_2-\tilde{L}_1}.
\end{equation}
Having found closed-form solutions for $R_{23}$ and $P_{23}$, we can begin to solve the outer problem in \eqref{decomposed_prob_general_case}, namely,
\begin{subequations}
\label{prob_slots_general_MAC}
\begin{align} 
\label{prob_slots_general_MAC_a}
\min_{P_{11}, P_{21},R_{21}} \quad & \tilde{L}_1 (P_{11} + P_{21})+ \bigl(\tfrac{\tilde{L}_2-\tilde{L}_1}{\alpha_2}\bigr) \bigl( 2^{\tfrac{B_2-\tilde{L}_1R_{21}}{\tilde{L}_2-\tilde{L}_1}}-1 \bigr)\\ 
\label{prob_slots_general_MAC_b}
\text{s.t.}   \quad \quad & \rev{ \eqref{Two_slots_genera_b}-\eqref{Two_slots_genera_d}, \eqref{Two_slots_genera_g},}\\
\label{prob_slots_general_MAC_c}
\quad & \changen{0 \le } \tfrac{B_2-\tilde{L}_1 R_{21}}{\tilde{L}_2-\tilde{L}_1} \le \log_2 (1+ \alpha_2 \bar{P}_{2}).
\end{align}
\end{subequations}
To do so, we decompose \eqref{prob_slots_general_MAC} as
\begin{eqnarray}
\label{decomposed_prob_general_case_MAC}
&\displaystyle{\min_{ R_{21}}}  \hspace{80pt} &\displaystyle{\min_{P_{11}, P_{21}}} \hspace{8pt}    \eqref{prob_slots_general_MAC_a} \quad \quad \\ 
&\quad \text{s.t.}   \hspace{8pt} \changen{\eqref{Two_slots_genera_c}},  \eqref{prob_slots_general_MAC_c},  \hspace{20pt} & \quad \text{s.t.} \hspace{8pt}   \rev{\eqref{Two_slots_genera_b}-\eqref{Two_slots_genera_d}},\eqref{Two_slots_genera_g}. \nonumber
\end{eqnarray} 

For a given rate $R_{21}$, the constraints in \eqref{Two_slots_genera_b}, \eqref{Two_slots_genera_c}, and \eqref{Two_slots_genera_g} construct a rectangular feasibility region of power pairs $(P_{11}, P_{21})$. Since the objective function in \eqref{prob_slots_general_MAC_a} is an increasing linear function of $P_{11}$ and $P_{21}$, it can be shown that at optimality the constraint in \eqref{Two_slots_genera_d} holds with equality. Let us define
\begin{subequations}
\label{lines_region}
\begin{align}
\label{ell_def}
\ell &=\{(P_{11},P_{21})|  \alpha_1 P_{11} + \alpha_2 P_{21}+1 = 2^{B_1/\tilde{L}_1+R_{21}} \}, \\
\ell_k &=\{(P_{11},P_{21})|   P_{ki} =(2^{R_{ki}}-1)/\alpha_k \}.
\end{align}
\end{subequations}
The line $\ell$ describes the power pairs for which the constraint in \eqref{Two_slots_genera_d} is active, and the lines $\ell_k$ describe the power pairs for which the constraint in \eqref{Two_slots_genera_b} and the right hand side of the constraint in \eqref{Two_slots_genera_c} are active\changen{, respectively}. Based on the intersection of the line $\ell $ with the rectangular region constructed by the constraints in \eqref{Two_slots_genera_b}, \eqref{Two_slots_genera_c} and \eqref{Two_slots_genera_g}, the feasibility region for the pair $(P_{11}, P_{21})$ will have different shapes; see Fig.~\ref{4_shapes}. If $\ell$ does not have any intersection with the rectangular region, the problem is not feasible.
\begin{figure}
  \begin{subfigure}[b]{0.5\linewidth}
    \centering
    \includegraphics[width=0.65\linewidth]{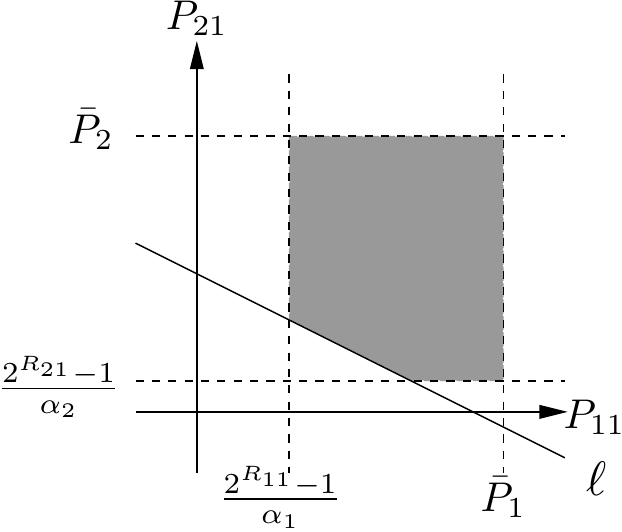} 
    \caption{Case I} 
    \label{case_I_fig} 
  \end{subfigure}
  \begin{subfigure}[b]{0.5\linewidth}
    \centering
    \includegraphics[width=0.65\linewidth]{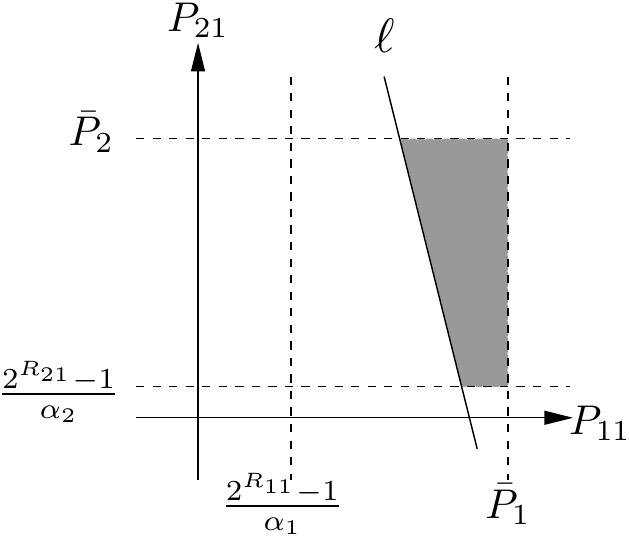} 
    \caption{Case II} 
    \label{case_II_fig} 
  \end{subfigure} 
  \begin{subfigure}[b]{0.5\linewidth}
    \centering
    \includegraphics[width=0.65\linewidth]{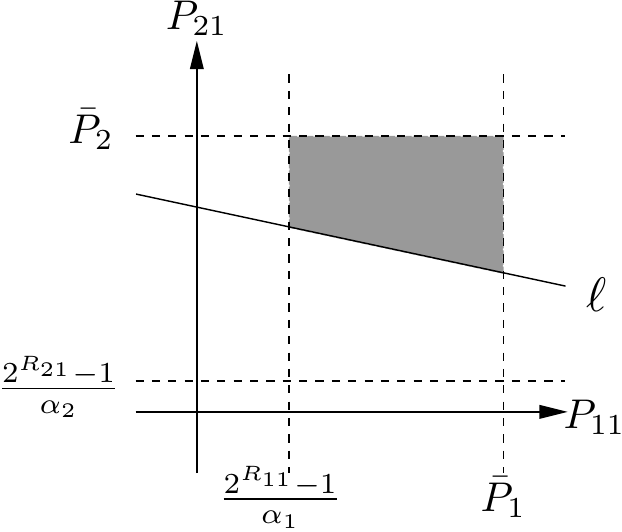} 
    \caption{Case III} 
    \label{case_III_fig} 
  \end{subfigure}
  \begin{subfigure}[b]{0.5\linewidth}
    \centering
    \includegraphics[width=0.65\linewidth]{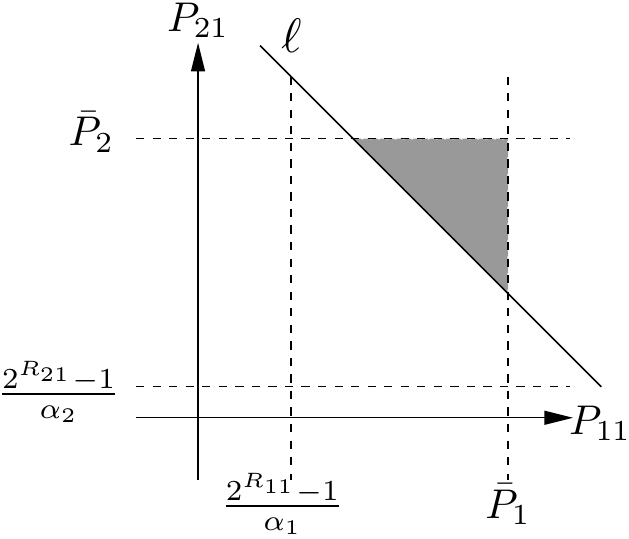} 
    \caption{Case IV} 
    \label{case_IV_fig} 
  \end{subfigure} 
  \caption{$(P_{11}, P_{21})$ feasibility region in different cases.}
  \label{4_shapes}
\end{figure}

According to the feasibility regions in Fig.~\ref{4_shapes} and the fact that the optimal values of $P_{11}$ and $P_{21}$ lie on the line $\ell$, the inner problem in \eqref{decomposed_prob_general_case_MAC} can be written, in terms of $P_{11}$, as
\begin{subequations}
\label{P_optimization_general_case_I}
\begin{eqnarray}
\label{P_optimization_general_case_I_a}
& \displaystyle{\min_{P_{11}}}  & \tilde{L}_1 \bigl((1-\tfrac{\alpha_1}{\alpha_2})P_{11} + \tfrac{2^{{B_1}/{\tilde{L}_1}+R_{21}}-1}{\alpha_2}\bigr)\\
\label{P_optimization_general_case_I_b}
&\text{s.t.}  & \eqref{Two_slots_genera_b}, \eqref{Two_slots_genera_c},\eqref{Two_slots_genera_g}.
\end{eqnarray}
\end{subequations}
\changen{The feasible set described by the constraints in \eqref{P_optimization_general_case_I_b} is the interval
\begin{equation}
\label{up_low_bnd}
\max \{ \tfrac{2^{B_1/\tilde{L}_1}-1}{\alpha_1}, \tfrac{2^{B_1/\tilde{L}_1+R_{21}}-1-\alpha_2\bar{P}_2}{\alpha_1}\} \le P_{11} \le \min\{\bar{P}_{1}, \tfrac{2^{R_{21}}(2^{B_1/\tilde{L}_1}-1)}{\alpha_1} \}, 
\end{equation}
and this interval is non-empty if and only if
\begin{subequations}
\begin{align}
\label{feas_constarint_P11}
&  \tfrac{B_1}{\tilde{L}_1} \le \log_2(1+\alpha_1\bar{P}_1), \\
\label{feas_constarint_R21}
&  R_{21} \le \min \{r_a, r_b\},
\end{align} 
\end{subequations}
where
\begin{subequations}
\begin{align}
\label{r_a}
&  r_a = \log_2(1+\alpha_2\bar{P}_2), \\
\label{r_b}
&  r_b = \log_2(1+\alpha_1\bar{P}_1+\alpha_2\bar{P}_2) - \tfrac{B_1}{\tilde{L}_1}.
\end{align} 
\end{subequations}
The constraint in \eqref{feas_constarint_P11} depends only on the parameters of the original problem in \eqref{prob_slots_general_K}, and the constraint in \eqref{feas_constarint_R21} will be incorporated into the various cases of the outer problem in \eqref{decomposed_prob_general_case_MAC}. }
Since the objective in \eqref{P_optimization_general_case_I_a} is linear in $P_{11}$, \changen{and the sign of the slope is the sign of $(1-\tfrac{\alpha_1}{\alpha_2})$, when the problem is feasible the parameter $\tfrac{\alpha_1}{\alpha_2}$ determines whether it is the upper bound or the lower bound on $P_{11}$ in \eqref{up_low_bnd} that is the optimal solution; i.e., }
\begin{subnumcases}{\hspace{-20pt}}
\label{caseI_A}
P_{11} = \max \{ \tfrac{2^{B_1/\tilde{L}_1}-1}{\alpha_1}, \tfrac{2^{B_1/\tilde{L}_1+R_{21}}-1-\alpha_2\bar{P}_2}{\alpha_1}\}, \quad  {\text{if}} ~ \tfrac{\alpha_1}{\alpha_2} \le 1, \\
\label{caseI_B}
 P_{11}= \min\{\bar{P}_{1}, \tfrac{2^{R_{21}}(2^{B_1/\tilde{L}_1}-1)}{\alpha_1} \}, \quad \quad \quad \hspace{26pt} {\text{if}} ~ \tfrac{\alpha_1}{\alpha_2} \ge 1.
\end{subnumcases}
\changen{Therefore, if the interval in \eqref{up_low_bnd} is not empty, then depending on whether $\tfrac{\alpha_1}{\alpha_2}$ is larger or smaller than one, we have one of the two pairs of candidate optimal solutions for $P_{11}$. These candidate pairs will be referred to as Cases A and B in the subsections below. Once the appropriate case is selected, for each candidate optimal solution for $P_{11}$ we obtain the corresponding $P_{21}$ using \eqref{ell_def}. Hence we obtain a candidate optimal solution to the inner problem in \eqref{decomposed_prob_general_case_MAC} and subsequently a formulation for the outer problem in \eqref{decomposed_prob_general_case_MAC}. These formulations will be referred to as subcases I and II of Cases A and B below. The optimal value for $R_{21}$ is then the feasible value that leads to the smaller objective out of subcases I and II.}

\subsection{Case A: $\tfrac{\alpha_1}{\alpha_2} \le 1$}
\label{Case_A_full_offloading}
Using \eqref{caseI_A} and \eqref{ell_def}, the candidate solutions of the inner problem when the channel gain of the second user is larger than that of the first user are
\begin{subnumcases}{\hspace{-20pt}}
\label{caseA_I}
P_{11} = \tfrac{2^{B_1/\tilde{L}_1}-1}{\alpha_1}, ~P_{21}= \tfrac{2^{B_1/\tilde{L}_1}(2^{R_{21}}-1)}{\alpha_2},\\
\label{caseA_II}
 P_{11}= \tfrac{2^{B_1/\tilde{L}_1+R_{21}}-\alpha_2\bar{P}_2-1}{\alpha_1},~ P_{21} = \bar{P}_2.
\end{subnumcases}
\subsubsection{Case A-I}
\label{case_A_I_full}
By substituting the values in \eqref{caseA_I} the outer optimization problem in \eqref{decomposed_prob_general_case_MAC} can be written as
\begin{subequations}
\label{prob_slots_general_Rs_K}
\begin{align} 
\label{prob_slots_general_Rs_K_a} 
\min_{R_{21}} \quad  &  \tilde{L}_1  \bigl( \tfrac{2^{B_1/\tilde{L}_1}-1}{\alpha_1}+ \tfrac{2^{B_1/\tilde{L}_1}(2^{R_{21}}-1)}{\alpha_2}\bigr) + \bigl(\tfrac{\tilde{L}_2-\tilde{L}_1}{\alpha_2}\bigr) \bigl( 2^{\frac{B_2-\tilde{L}_1R_{21}}{\tilde{L}_2-\tilde{L}_1}}-1 \bigr) \\ 
\label{prob_slots_general_Rs_K_b}
\text{s.t.}   \quad & \max \{0, r_c\}  \le R_{21} \le  \min \{r_a, r_b \},
\end{align}
\end{subequations}
\changen{where 
\begin{equation}
\label{R_c_constraint}
r_c = \tfrac{B_2- (\tilde{L}_2-\tilde{L}_1)\log_2 (1+ \alpha_2 \bar{P}_{2})}{ \tilde{L}_1}.
\end{equation}
We observe that the constraints in \eqref{prob_slots_general_Rs_K_b} depend only on the parameters of the original problem in \eqref{prob_slots_general_K}. }
In Appendix~\ref{app_cvx_R21_case_I} we show that the objective function in \eqref{prob_slots_general_Rs_K} is a convex function of $R_{21}$. Since the constraints are simple bounds on $R_{21}$, the optimization problem in \eqref{prob_slots_general_Rs_K} is convex and there are three possible values for the optimal value of $R_{21}$. If it is feasible, the value for which the first derivative of the objective function is equal to zero, $R_{{21}_{\text{d}}} = \tfrac{B_2}{\tilde{L}_2} - \tfrac{(\tilde{L}_2-\tilde{L}_1)B_1}{\tilde{L}_1 \tilde{L}_2}$, 
is the optimal solution. Otherwise, either the upper bound or the lower bound on $R_{21}$ is the optimal solution, or the problem is infeasible.
\subsubsection{Case A-II}
\label{case_A_II_full}
Using the transmission powers in \eqref{caseA_II}, the rate optimization problem can be written as
\begin{subequations}
\label{prob_slots_general_Rs_K_case_II}
\begin{align}
\label{prob_slots_general_Rs_K_case_II_a} 
\min_{ R_{21}} \quad &  \tilde{L}_1 \bigl(\tfrac{2^{B_1/\tilde{L}_1+R_{21}}-\alpha_2 \bar{P}_2-1}{\alpha_1}+ \bar{P}_2\bigr)+ \bigl(\tfrac{\tilde{L}_2-\tilde{L}_1}{\alpha_2}\bigr) \bigl(2^{\frac{B_2-\tilde{L}_1R_{21}}{\tilde{L}_2-\tilde{L}_1}}-1\bigr) \\ 
\label{prob_slots_general_Rs_K_case_II_b}
\text{s.t.}  \quad & \eqref{prob_slots_general_Rs_K_b}.
\end{align}
\end{subequations}
The objective function in \eqref{prob_slots_general_Rs_K_case_II} is convex in terms of $R_{21}$ and the constraints are simple bounds. Hence, the optimal value for $R_{21}$ can be at either end point of its feasibility interval or at the point at which the derivative of the objective is equal to zero.
\subsection{Case B: $\tfrac{\alpha_1}{\alpha_2} \ge 1$}
Using \eqref{caseI_B} and \eqref{ell_def}, the candidate solutions of the inner problem when the channel gain of the first user is greater than that of the second user are
\begin{subnumcases}{\hspace{-20pt}}
\label{caseB_I}
 P_{11}= \tfrac{2^{R_{21}}(2^{B_1/\tilde{L}_1}-1)}{\alpha_1},~ P_{21} = \tfrac{2^{R_{21}}-1}{\alpha_2},\\
\label{caseB_II}
P_{11} =\bar{P}_1, ~P_{21}= \tfrac{2^{B_1/\tilde{L}_1+R_{21}}-\alpha_1 \bar{P}_1-1}{\alpha_2}.
\end{subnumcases}
\subsubsection{Case B-I}
\label{case_B_I_full}
By substituting the values in \eqref{caseB_I}, the outer problem in \eqref{decomposed_prob_general_case_MAC} becomes
\begin{subequations}
\label{prob_slots_general_Rs_K_alpha12}
\begin{align} 
\label{prob_slots_general_Rs_K_alpha12_a} 
\min_{R_{21}} \quad  & \tilde{L}_1  \bigl( \tfrac{2^{B_1/\tilde{L}_1 + R_{21}}-2^{R_{21}}}{\alpha_1}+ \tfrac{2^{R_{21}}-1}{\alpha_2}\bigr) + \bigl(\tfrac{\tilde{L}_2-\tilde{L}_1}{\alpha_2}\bigr) \bigl( 2^{\frac{B_2-\tilde{L}_1R_{21}}{\tilde{L}_2-\tilde{L}_1}}-1 \bigr) \\ 
\label{prob_slots_general_Rs_K_alpha12_b}
\text{s.t.}  \quad  & \eqref{prob_slots_general_Rs_K_b}.
\end{align}
\end{subequations}
Applying the techniques used in Appendix~\ref{app_cvx_R21_case_I}, it can be shown that the objective function in \eqref{prob_slots_general_Rs_K_alpha12} is a convex function of $R_{21}$. Since the constraints in \eqref{prob_slots_general_Rs_K_alpha12} are the same as those in \eqref{prob_slots_general_Rs_K}, there are three possible values for the optimal value of $R_{21}$. If it is feasible, the value for which the derivative of the objective is zero, $R_{{21}_{\text{d}}} = \tfrac{B_2}{\tilde{L}_2} - \tfrac{(\tilde{L}_2-\tilde{L}_1) \phi_1}{\tilde{L}_2},$
where $\phi_1 = \log_2 \bigl(\tfrac{\alpha_2}{\alpha_1} (2^{B_1/\tilde{L}_1}-1) +1\bigr)$, is the optimal solution. Otherwise, either the upper bound or the lower bound on $R_{21}$ is the optimal solution, or the problem is infeasible.
\subsubsection{Case B-II}
\label{case_B_II_full}
By substituting the values in \eqref{caseB_II}, the rate optimization problem for this candidate solution is 
\begin{subequations}
\label{prob_slots_general_Rs_K_case_III}
\begin{align} 
\label{prob_slots_general_Rs_K_case_III_a}
\min_{R_{21}} \quad & \tilde{L}_1 \bigl(\tfrac{2^{B_1/\tilde{L}_1+R_{21}}-\alpha_1 \bar{P}_1-1}{\alpha_2}+ \bar{P}_1\bigr) + \bigl(\tfrac{\tilde{L}_2-\tilde{L}_1}{\alpha_2}\bigr) \bigl(2^{\frac{B_2-\tilde{L}_1R_{21}}{\tilde{L}_2-\tilde{L}_1}}-1\bigr) \\ 
\label{prob_slots_general_Rs_K_case_III_b}
\text{s.t.}  \quad  &  \eqref{prob_slots_general_Rs_K_b}.
\end{align}
\end{subequations}
The objective function in \eqref{prob_slots_general_Rs_K_case_III} has the same structure as the objective function in \eqref{prob_slots_general_Rs_K_case_II}. Accordingly, the optimal solution for the transmission rate $R_{21}$ is either at one of the end points of the feasible interval, or at the point in which the derivative of the objective function is equal to zero, if that point lies within the feasibility region. 

\subsection{Algorithm for Solving \eqref{prob_slots_general_K}}
By considering the derivations in the subsections above, we obtain a closed-form expression for an optimal solution to \eqref{prob_slots_general_K}. The steps that need to be taken to obtain this solution are summarized in Algorithm~\ref{main_algorithm}. \changen{The first step is to determine if the problem in \eqref{prob_slots_general_K} is feasible. A necessary and sufficient condition for feasibility is that \eqref{feas_constarint_P11} holds and that there exists at least one $R_{21}$ that satisfies \eqref{prob_slots_general_Rs_K_b}. Both of those expressions depend on the parameters of the problem. If the problem is feasible, then depending on the value of $\tfrac{\alpha_1}{\alpha_2}$ we either solve the two subcases of Case A or the two subcases of Case B. The optimal solution is the better of those two candidate solutions.}


 \begin{algorithm}[htp]
\caption{: An optimal solution to \eqref{prob_slots_general_K}} 
\label{main_algorithm}
\begin{algorithmic}
\State \changen{Input data: values of $\{ B_k \} $, $\{ \bar{P}_k \}$, $\{ L_k \}$, $\{ T_k \}$, $\{ \alpha_k \}$, and $T_s$.
	\If {$\log_2(1+\alpha_1 \bar{P}_1) <  \tfrac{B_1}{\tilde{L}_1} $ or $ \max \{0, r_c\}  >  \min \{r_a, r_b \}$ }
	\State The problem is infeasible.
	\Else 
	\If {$\tfrac{\alpha_1}{\alpha_2} \le 1$} 
     		\State Generate a partial candidate solution according to Case A-I (Section~\ref{case_A_I_full}).	
     		That is, let $R_{21}^{\star}$ denote the optimal solution to \eqref{prob_slots_general_Rs_K}, and calculate $P_{11}^{\star}$ and $P_{21}^{\star}$ using \eqref{caseA_I}.    		
		\State Generate a partial candidate solution according to Case A-II (Section~\ref{case_A_II_full}).
		That is, let $R_{21}^{\star}$ denote the optimal solution to \eqref{prob_slots_general_Rs_K_case_II}, and calculate $P_{11}^{\star}$ and $P_{21}^{\star}$ using \eqref{caseA_II}.
		
	\Else  
     		\State Generate a partial candidate solution according to Case B-I (Section~\ref{case_B_I_full}).	
     		That is, let $R_{21}^{\star}$ denote the optimal solution to \eqref{prob_slots_general_Rs_K_alpha12}, and calculate $P_{11}^{\star}$ and $P_{21}^{\star}$ using \eqref{caseB_I}.
     		
		\State Generate a partial candidate solution according to Case B-II (Section~\ref{case_B_II_full}).
		That is, let $R_{21}^{\star}$ denote the optimal solution to \eqref{prob_slots_general_Rs_K_case_III}, and calculate $P_{11}^{\star}$ and $P_{21}^{\star}$ using \eqref{caseB_II}.
	\EndIf 
    	\State Complete each partial candidate solution by choosing $R_{23}^{\star}$ according to \eqref{Opt_R_23}, $P_{23}^{\star}$ according to \eqref{Opt_P_23}, and setting $\tau_1= \tilde{L}_1$, $P_{12} = 0$, $R_{12}=0$ and $R_{11} = \tfrac{B_1}{\tau_1}$. Calculate the objective value for each candidate solution, and choose the solution that corresponds to the minimum value. 
			\EndIf} 
\end{algorithmic}
\end{algorithm}

\section{Complete Computation Offloading: Suboptimal Multiple Access Methods}
\label{suboptimal_schemes}
Having obtained a closed-form expression for an optimal solution to the minimum energy offloading problem when the full capabilities of the multiple access channel are employed, we now address that problem when suboptimal multiple access methods are employed. That is, we will solve the variant of the problem in \eqref{prob_slots_main} in which the achievable rate region constraint of the first time slot, cf. \eqref{prob_slots_main_h}, is the rate region of the chosen suboptimal scheme rather than the capacity region. We will consider the cases of TDMA, sequential decoding without time sharing (SDwts), and independent decoding (ID).  
\subsection{Time Division Multiple Access}
\label{TDMA}
In the TDMA scheme only one user is transmitting at a time and the communication resource is fully assigned to that user. Quite naturally, that means that the optimal energy consumption can be achieved using only two time slots. By simplifying our notation in an intuitive way and observing that the durations of the time slots are $\tfrac{B_1}{R_1}$ and $\tfrac{B_2}{R_2}$, respectively, the optimization problem in the TDMA case becomes 
\begin{subequations}
\label{TDMA_prob}
\begin{align}
\min_{\substack{P_1, P_2,\\ R_1, R_2}} \quad & \bigl(  \tfrac{B_1 P_1}{R_1} + \tfrac{B_2 P_2}{R_2} \bigr)\\ 
\label{TDMA_prob_b}
\text{s.t.}  \enspace  \quad  & 0 \le R_k \le \log_2(1+\alpha_k P_k), \quad k=1, 2, \\
  \label{TDMA_prob_c}
& 0 \le P_k \le \bar{P}_k, \quad k=1, 2,  \\
 \label{TDMA_prob_d}
&  \tfrac{B_1}{R_1} \le \tilde{L}_1,\\
 \label{TDMA_prob_e}
&  \tfrac{B_1}{R_1}+\tfrac{B_2}{R_2} \le \tilde{L}_2.
\end{align}
\end{subequations}

The optimal powers in \eqref{TDMA_prob} are those that achieve equality in the upper bounds in \eqref{TDMA_prob_b}, and hence the problem can be simplified to 
\begin{subequations}
\label{TDMA_rates_prob}
\begin{align}
\label{TDMA_rates_prob_a}
\min_{ R_1, R_2} \quad & \bigl( \tfrac{B_1}{R_1}\bigr) \bigl(\tfrac{2^{R_1}-1}{\alpha_1} \bigr) + \bigl(\tfrac{B_2}{R_2}\bigr) \bigl(\tfrac{2^{R_2}-1}{\alpha_2}\bigr) \\ 
\label{TDMA_rates_prob_b}
\text{s.t.}  \quad  & 0 \le R_1 \le \log_2(1+\alpha_1 \bar{P}_1), \\
 \label{TDMA_rates_prob_c}
  \quad  & 0 \le R_2 \le \log_2(1+\alpha_2 \bar{P}_2), \\
   \label{TDMA_rates_prob_d}
  \quad &  \rev{\eqref{TDMA_prob_d} \text{ and } \eqref{TDMA_prob_e} }.
\end{align}
\end{subequations}
The problem in \eqref{TDMA_rates_prob} can be decomposed as
\begin{eqnarray}
\label{decomposed_prob_rates}
& \displaystyle{\min_{R_1}}  \hspace{65pt} & \displaystyle{\min_{R_2}} \hspace{8pt} \eqref{TDMA_rates_prob_a} \\ 
&\text{s.t.}  \hspace{8pt}   \eqref{TDMA_prob_d}, \eqref{TDMA_rates_prob_b},  & \text{s.t.} \hspace{8pt} \eqref{TDMA_prob_e}, \eqref{TDMA_rates_prob_c}. \nonumber
\end{eqnarray}
The objective function of the problem in \eqref{decomposed_prob_rates} is increasing in terms of $R_2$ and the optimal transmission rate for the second user is achieved when \rev{\eqref{TDMA_prob_e}} holds with equality, i.e., $R_2 = \tfrac{B_2 R_1}{\tilde{L}_2 R_1 - B_1}$. The outer optimization problem is then
\begin{subequations}
\label{TDMA_final_prob}
\begin{align}
\min_{R_1} \quad &  \bigl(\tfrac{B_1}{R_1}\bigr) \bigl(\tfrac{2^{R_1}-1}{\alpha_1}\bigr) + \bigl(\tfrac{\tilde{L}_2 R_1- B_1}{\alpha_2 R_1}\bigr) \bigl(2^{\frac{B_2 R_1 }{\tilde{L}_2 R_1- B_1}}-1\bigr)\\ 
\label{TDMA_final_prob_b}
\text{s.t.}  \quad  & \rev{\eqref{TDMA_prob_d} \text{ and } \eqref{TDMA_rates_prob_b}}, \\
 \label{TDMA_final_prob_c}
  \quad  & \tfrac{B_2 R_1}{\tilde{L}_2 R_1 - B_1} \le \log_2(1+\alpha_2 \bar{P}_2).
\end{align}
\end{subequations}
It is shown in \cite{Mahsaasilomar2017} that the objective in \eqref{TDMA_final_prob} is convex. Since the constraints are linear, the optimal solution is either where the derivative of the objective is zero or at one of the end points of the feasibility interval. For each of those values for $R_1$ the corresponding values for $R_2$, $P_1$ and $P_2$ can be obtained and the quadruple that provides the smallest objective value in \eqref{TDMA_prob} is the optimal solution.
\subsection{Sequential Decoding without time sharing}
\label{Sec_Dec}
In the sequential decoding without time sharing scheme, the received signal from one user is decoded considering the interference from the other user as noise. Presuming that this message is correctly decoded, the interference of the decoded user is then reconstructed and subtracted from the received signal and the other user is decoded without interference. We will look at the case in which the system can choose the order in which the users are decoded, but that order remains fixed for the duration of the multiple access interval. Hence, we use the qualifier ``without time sharing" (wts) in our description. The achievable rate region of sequential decoding (wts) is shown in Fig.~\ref{feas_region_all_sec_dec}. For this scheme there are scenarios in which the optimal energy consumption requires that all three time slots of the system be employed. Hence, the problem of finding the minimum energy consumption for this scheme is 
\begin{subequations}
\label{prob_sec_dec}
\begin{align} 
\label{prob_sec_dec_a}
\min_{\substack {P_{11}, P_{12}, P_{21}, P_{23}, \\ R_{11}, R_{12}, R_{21}, R_{23}, \tau_1}} \quad & \tau_1 (P_{11} + P_{21})+\bigl(\tfrac{B_1-\tau_1R_{11}}{R_{12}}\bigr) P_{12} + \bigl(\tfrac{B_2-\tau_1R_{21}}{R_{23}}\bigr) P_{23} \\ 
\label{prob_sec_dec_b}
\text{s.t.}  \enspace \quad \quad \quad & \eqref{prob_slots_main_b1}-\eqref{prob_slots_main_g},\\
\label{prob_sec_dec_c}
\quad & \{ R_{11}, R_{21} \} \in \tilde{\mathcal{R}}_1 \cup  \tilde{\mathcal{R}}_2, 
\end{align}
\end{subequations}
where
\begin{subequations}
\begin{align} \nonumber
 \tilde{\mathcal{R}}_1 = & \bigl\{ \{ R_{11}, R_{21} \} |  0 \le R_{11}  \le \log_2(1+ \alpha_1 P_{11}), \quad  \\ 
 & \quad \quad \quad \quad \quad \hspace{5pt} 0 \le R_{21}  \le \log_2(1+\tfrac{\alpha_2 P_{21}}{1+\alpha_1 P_{11}})  \bigr\},  \\ \nonumber
\tilde{\mathcal{R}}_2 = & \bigl\{ \{ R_{11}, R_{21} \} |  0 \le R_{11}  \le \log_2(1+\tfrac{\alpha_1 P_{11}}{1+\alpha_2 P_{21}}), \quad \\
 & \quad \quad \quad \quad \quad \hspace{5pt} 0 \le R_{21}  \le  \log_2(1+ \alpha_2 P_{21})  \bigr\}.
 \end{align}
\end{subequations}
The problem in \eqref{prob_sec_dec} can be decomposed and a closed-form solution for the transmission rates and the transmission powers of the second and third time slots can be obtained in terms of the transmission rates and transmission powers of the first time slot \cite{Mahsaasilomar2017}. The remaining optimization problem is
\begin{subequations}
\label{prob_sec_dec_out}
\begin{align}
\label{prob_sec_dec_out_a}
\min_{\substack {P_{11}, P_{21}, \\ R_{11}, R_{21}, \tau_1}} \quad & \tau_1 (P_{11} + P_{21})+ \bigl(\tfrac{\tilde{L}_1-\tau_1}{\alpha_1}\bigr) \bigl(2^{\frac{B_1-\tau_1R_{11}}{\tilde{L}_1-\tau_1}} -1 \bigr) + \bigl(\tfrac{\tilde{L}_2-\tilde{L}_1}{\alpha_2}\bigr) \bigl(2^{\frac{B_2-\tau_1R_{21}}{\tilde{L}_2-\tilde{L}_1}} -1 \bigr) \\ 
\label{prob_sec_dec_out_b}
\text{s.t.}  \enspace \quad \quad  & \eqref{prob_slots_main_b1}, \eqref{prob_slots_main_b}\\
\label{prob_sec_dec_out_c}
& \tfrac{B_1-\tau_1R_{11}}{\tilde{L}_1-\tau_1} \le \log_2(1+\alpha_1 \bar{P}_1), \\
\label{prob_sec_dec_out_d}
 & \tfrac{B_2-\tau_1R_{21}}{\tilde{L}_2-\tilde{L}_1} \le \log_2(1+\alpha_2 \bar{P}_2),\\
\label{prob_sec_dec_out_e}
\quad & \{ R_{11}, R_{21} \} \in \tilde{\mathcal{R}}_1 \cup \tilde{\mathcal{R}}_2.
\end{align}
\end{subequations}

Depending on the corner point at which sequential decoding scheme is operating, a closed-form solution for the powers can be obtained in terms of the rates \cite{Mahsaasilomar2017}, namely,
\[ \begin{cases}
    & P_{11} = \tfrac{2^{R_{11}}-1}{\alpha_1}, \quad P_{21} = \tfrac{2^{R_{11}} (2^{R_{21}}-1)}{\alpha_2} \quad \text{for } \tilde{\mathcal{R}}_1\\
    & P_{11} = \tfrac{2^{R_{21}} (2^{R_{11}}-1)}{\alpha_1}, \quad P_{21} = \tfrac{2^{R_{21}}-1}{\alpha_2} \quad \text{for } \tilde{\mathcal{R}}_2
  \end{cases}
\]
Since the rate regions for each decoding order, $\tilde{\mathcal{R}}_1$ and $\tilde{\mathcal{R}}_2$, are rectangular \rev{(see Fig.~\ref{feas_region_all_sec_dec})}, the optimal rate pair lies at the ``dominant'' corner \rev{(i.e., the ``North-East'' corner)} of one of \rev{the} rectangles. \rev{To} determine \rev{which rectangle, and hence} the optimal triple $(R_{11}, R_{21}, \tau_1)$ \rev{in \eqref{prob_sec_dec_out}, we will first determine the optimal triple} for each decoding order and then select the triple that generates the lower energy solution. The optimization problem for $(R_{11}, R_{21}, \tau_1)$ has a similar structure in each case, and in the case of $\tilde{\mathcal{R}}_1$ it is
\begin{subequations}
\label{prob_sec_dec_out_Rs}
\begin{align} 
\label{prob_sec_dec_out_Rs_a}
\min_{\substack {R_{11}, R_{21}, \tau_1}} \quad &   \tau_1 \bigl(\tfrac{2^{R_{11}}-1}{\alpha_1} +  \tfrac{2^{R_{11}} (2^{R_{21}}-1)}{\alpha_2} \bigr)  + \bigl(\tfrac{\tilde{L}_1-\tau_1}{\alpha_1}\bigr) \bigl(2^{\frac{B_1-\tau_1R_{11}}{\tilde{L}_1-\tau_1}} -1 \bigr) \\
& \quad \quad + \bigl(\tfrac{\tilde{L}_2-\tilde{L}_1}{\alpha_2}\bigr) \bigl(2^{\frac{B_2-\tau_1R_{21}}{\tilde{L}_2-\tilde{L}_1}} -1 \bigr) \\ 
\label{prob_sec_dec_out_Rs_b}
\text{s.t.}   \quad  \quad & \eqref{prob_slots_main_b1}, \eqref{prob_slots_main_b},\\
\label{prob_sec_dec_out_Rs_c}
\quad  & \tfrac{2^{R_{11}}-1}{\alpha_1} \le \bar{P}_1,\\
\label{prob_sec_dec_out_Rs_d}
\quad & \tfrac{2^{R_{11}} (2^{R_{21}}-1)}{\alpha_2} \le \bar{P}_2, \\
\label{prob_sec_dec_out_Rs_e}
& \rev{\eqref{prob_sec_dec_out_c} \text{ and } \eqref{prob_sec_dec_out_d}}.
\end{align}
\end{subequations}
\rev{This problem is not known to be convex, but a variety of solution strategies can be developed, including algorithms based on augmented Lagrangian techniques \cite[Chapter~17]{nocedal2006numerical}, sequential quadratic programming \cite[Chapter~18]{nocedal2006numerical}, \cite{lawrence2001computationally}, and successive convex aproximation \cite{razaviyayn2013unified, scutari2017parallel}. In this manuscript we will employ a simpler strategy that is based on the observation that when} any two of the variables in \eqref{prob_sec_dec_out_Rs} are fixed, the objective function is convex in the remaining variable and the constraints can be written so that they are linear in that variable. \rev{That} observation \rev{ suggests the adoption of} a coordinate descent method for solving \eqref{prob_sec_dec_out_Rs}. \rev{The convexity of \eqref{prob_sec_dec_out_Rs} in each coordinate (alone), and other properties of the objective and the constraints, enable us to show that the coordinate descent method} is guaranteed to converge to a stationary solution of \eqref{prob_sec_dec_out_Rs} \cite[Theorem 1]{hong2016unified}. In \rev{all} our numerical experiments, \rev{only some of which are shown in Section~\ref{simulation}, our coordinate descent method} converged to the globally optimal solution.
\subsection{Independent Decoding}
\label{Ind_Dec}
In the independent decoding scheme, the received signals of both users are decoded independently, with the interference of the other user being considered as noise. The achievable rate region in this scheme is depicted in Fig.~\ref{feas_region_all_ind}. As in the case of sequential decoding (wts), the optimal solution may activate all three time slots. Therefore, the minimum energy consumption optimization problem in the independent decoding case \rev{can be written as}
\begin{subequations}
\label{prob_ind_dec}
\begin{align} 
\label{prob_ind_dec_a}
\min_{\substack {P_{11}, P_{12}, P_{21}, P_{23}, \\ R_{11}, R_{12}, R_{21}, R_{23}, \tau_1}} \quad &  \tau_1 (P_{11} + P_{21})+ \bigl(\tfrac{B_1-\tau_1R_{11}}{R_{12}}\bigr) P_{12} + \bigl(\tfrac{B_2-\tau_1R_{21}}{R_{23}}\bigr) P_{23} \\ 
\label{prob_ind_dec_b}
\text{s.t.}  \enspace \quad \quad  \quad & \eqref{prob_slots_main_b1}-\eqref{prob_slots_main_g},\\
\label{prob_ind_dec_c}
\quad  & 0 \le R_{11} \le \log_2 (1+ \tfrac{\alpha_1 P_{11}}{1+\alpha_2 P_{21}}),\\
\label{prob_ind_dec_d}
\quad & 0 \le R_{21} \le \log_2 (1+ \tfrac{\alpha_2 P_{21}}{1+\alpha_1 P_{11}}), 
\end{align}
\end{subequations}
where \eqref{prob_ind_dec_c} and \eqref{prob_ind_dec_d} \rev{describe} the achievable rate \rev{region} in the first time slot. In terms of the variables of the second and third time slots, the problem in \eqref{prob_ind_dec} is similar to that in \eqref{prob_sec_dec} for the sequential decoding (wts) case. Accordingly, the optimal transmission rates and transmission powers of these slots can be obtained in terms of the rates and powers of the first time slot (cf.\ Section~\ref{Sec_Dec}), and the problem can be simplified to
\begin{subequations}
\label{prob_ind_dec_out}
\begin{align} \notag
\label{prob_ind_dec_out_a}
\min_{\substack {P_{11}, P_{21}, \\ R_{11}, R_{21}, \tau_1}} \quad & \tau_1 (P_{11} + P_{21})+\bigl(\tfrac{\tilde{L}_1-\tau_1}{\alpha_1}\bigr) \bigl(2^{\tfrac{B_1-\tau_1R_{11}}{\tilde{L}_1-\tau_1}} -1 \bigr)  \\
& \quad + \bigl(\tfrac{\tilde{L}_2-\tilde{L}_1}{\alpha_2}\bigr) \bigl(2^{\tfrac{B_2-\tau_1R_{21}}{\tilde{L}_2-\tilde{L}_1}} -1\bigr) \\ 
\label{prob_ind_dec_out_b}
\text{s.t.}   \quad \quad  & \eqref{prob_slots_main_b1}-\eqref{prob_slots_main_b_new},\eqref{prob_slots_main_e}, \rev{\eqref{prob_sec_dec_out_c}, \eqref{prob_sec_dec_out_d}, \eqref{prob_ind_dec_c}, \eqref{prob_ind_dec_d}}.
\end{align}
\end{subequations}
By decomposing the problem in \eqref{prob_ind_dec_out}, the powers of the first time slot can be obtained in terms of the rates of that slot \cite{Mahsaasilomar2017}, 
\begin{align}
\label{Ind_Dec_powers_rates}
P_{11}&=\tfrac{2^{R_{21}}(2^{R_{11}}-1)}{\alpha_1 (2^{R_{11}}+2^{R_{21}}-2^{R_{11}+R_{21}})},\\
P_{21} &=\tfrac{2^{R_{11}}(2^{R_{21}}-1)}{\alpha_2 (2^{R_{11}}+2^{R_{21}}-2^{R_{11}+R_{21}})}.
\end{align}
Accordingly, the \rev{remaining} optimization problem \rev{reduces to}
\begin{subequations}
\label{prob_ind_dec_out_Rs}
\begin{align}
\label{prob_ind_dec_out_Rs_a}
\min_{\substack {R_{11}, R_{21}, \tau_1}} \quad &  f(R_{11}, R_{21}, \tau_1) \\ 
\label{prob_ind_dec_out_Rs_b}
\text{s.t.}  \quad \quad  & \eqref{prob_slots_main_b1}-\eqref{prob_slots_main_b_new}, \rev{\eqref{prob_sec_dec_out_c}, \eqref{prob_sec_dec_out_d},}\\
\label{prob_ind_dec_out_Rs_e}
\quad  & \tfrac{2^{R_{21}}(2^{R_{11}}-1)}{\alpha_1 (2^{R_{11}}+2^{R_{21}}-2^{R_{11}+R_{21}})} \le \bar{P}_1,\\
\label{prob_ind_dec_out_Rs_f}
\quad & \tfrac{2^{R_{11}}(2^{R_{21}}-1)} { \alpha_2 (2^{R_{11}}+2^{R_{21}}-2^{R_{11}+R_{21}})} \le \bar{P}_2,
\end{align}
\end{subequations}
where $f(R_{11}, R_{21}, \tau_1)$ is shown  in \eqref{f_function} at the top of this page.

Similar to the problem in \eqref{prob_sec_dec_out_Rs}, we can show that the objective function in \eqref{prob_ind_dec_out_Rs} is convex in each of the variables when the other two are given, and the constraints can be written so that they are linear in the corresponding variable. Hence, we will adopt a coordinate descent approach to solving \eqref{prob_ind_dec_out_Rs}. It can be shown that that approach is guaranteed to converge to a stationary solution, and in \rev{all} our numerical experiments it converged to the globally optimal solution.  
\begin{figure*}[tp]
    \begin{equation}
\begin{aligned}
\label{f_function}
f(R_{11}, R_{21}, \tau_1) = & \tau_1 \bigl(\tfrac{2^{R_{21}}(2^{R_{11}}-1)}{\alpha_1 (2^{R_{11}}+2^{R_{21}}-2^{R_{11}+R_{21}})}+ \tfrac{2^{R_{11}}(2^{R_{21}}-1)} { \alpha_2 (2^{R_{11}}+2^{R_{21}}-2^{R_{11}+R_{21}})} \bigr) \\
 & \hspace{20pt} +  \bigl(\tfrac{\tilde{L}_1-\tau_1}{\alpha_1}\bigr) \bigl(2^{\frac{B_1-\tau_1R_{11}}{\tilde{L}_1-\tau_1}} -1 \bigr) + \bigl(\tfrac{\tilde{L}_2-\tilde{L}_1}{\alpha_2}\bigr) \bigl(2^{\frac{B_2-\tau_1R_{21}}{\tilde{L}_2-\tilde{L}_1}} -1 \bigr).
\end{aligned}
    \end{equation}
    \noindent\rule{18.2cm}{0.4pt}
    \end{figure*}
    
\rev{\section{On the Choice of the Multiple Access Scheme for Complete Computation Offloading}
\label{insight}
The closed-form expressions that we have obtained for the optimal communication resource allocation in the case of complete computation offloading with the full multiple access scheme and TDMA, and the quasi-closed-form expressions that we have obtained for the cases of independent decoding and sequential decoding (without time sharing), enable us to gain insight into the impact of the choice of the multiple access scheme. 

The first result is that whenever TDMA produces a solution that is feasible for the offloading problem, that solution is also an optimal solution for the independent decoding case. (This is consistent with the fact that in the TDMA case the decoders work independently.) As one might expect, there are scenarios in which a three-time-slot independent decoding scheme provides a feasible solution to the offloading problem, but TDMA does not, and we will provide some examples of such scenarios in Section~\ref{simulation_full}. However, whenever TDMA is feasible it is optimal for the independent decoding case, and in certain circumstances it will have a more straightforward implementation. We formally state this property in the following proposition.

\begin{prop}
\label{prob_1}
If the problem in \eqref{TDMA_final_prob} is feasible, let $P_k^{\star}$ and $R_k^{\star}$ denote the solution to the problem in \eqref{TDMA_prob} that is derived in Section~\ref{TDMA}. In that case, an optimal solution to the problem in \eqref{prob_ind_dec} is $P_{11}=P_{21}=0, P_{12} = P_1^{\star}, P_{23} = P_2^{\star}, R_{11}=R_{21}=0, R_{12} = R_1^{\star}, R_{23} = R_2^{\star}$ and $\tau_1 = 0$.
\end{prop}
\begin{proof}
See Appendix~\ref{TDMA_opt_IndDec}.
\end{proof}

Our second result shows that when the channel gains of both users are equal and the power budgets are above an explicit threshold then the optimal resource allocation for the TDMA scheme reduces the energy consumption to the same level as the optimal resource allocation for the full multiple access scheme. In other words, when the channel gains are equal, simplifying the implementation by constraining the multiple access scheme to be TDMA does not result in loss of optimality (for sufficiently large power budgets). Having said that, as we will show in Section~\ref{simulation}, when the channel gains are significantly different, exploiting the full capabilities of the multiple access channel enables substantial reduction in the energy required to offload the tasks. The optimality of the TDMA scheme is formalized in the following proposition. 

\begin{prop}
\label{prob_2}
Let $P_{ki}^{\star}$ denote an optimal solution for $P_{ki}$ in the problem in \eqref{prob_slots_general_K} when $|h_1|^2 = |h_2|^2$. If $P_{11}^{\star}+P_{21}^{\star} \le \min \{ \bar{P}_1, \bar{P}_2 \}$, then TDMA can obtain the optimal energy consumption of the full multiple access scheme.
\end{prop}
\begin{proof}
See Appendix~\ref{app_equal_channel_gain}.
\end{proof}

Since in a time-slotted system the optimal TDMA scheme is an optimal independent decoding scheme whenever it is feasible (Proposition~\ref{prob_1}), a consequence of Proposition~\ref{prob_2} is that, for power budgets above the threshold, independent decoding is also optimal when the channel gains are the same. Since the achievable rate region of the sequential decoding (wts) scheme is no smaller than that of independent decoding, Proposition~\ref{prob_2} also implies that sequential decoding (wts) is optimal when the channel gains are equal. }

\section{Partial Computation Offloading}
\label{partial_offloading}
Up until this point, we have considered indivisible computational tasks that are either completely offloaded or executed locally. For divisible computational tasks, we have the opportunity to take advantage of the implicit parallelism of the mobile station and the access point by offloading a portion of the computational task to the access point, with the remainder being executed locally.

Our first observation in the development of resource allocation algorithms for the partial offloading case is that the transmission energy and the communication latency associated with offloading a portion of the task depend on its description length, whereas the computational energy and latency associated with executing the remaining portion locally depend on the number of operations required. In this section we will focus on the class of ``data-partitionable'' tasks \cite{wang2016mobile}. Such tasks involve a relatively simple-to-describe action being applied, independently, to multiple blocks of data. As such, the number of operations required to complete a fraction of the task can be modeled as being a function of the description length \cite{wang2016mobile, munoz2015optimization, zhang2013energy}. For simplicity we will consider the limiting case in which the tasks can be partitioned finely enough that the partition can be modeled by a continuous variable. 

In the generic scenario of partial offloading, both users will be offloading a portion of their tasks and the time slotted communication model in Fig.~\ref{time_slots} applies. (Note that without loss of generality we have ordered the users so that $L_1 \le L_2$.) Based on the outcomes of the indivisible case, we will focus on the full multiple access and TDMA schemes, and hence, we need only consider two of the time slots ($\tau_2 = 0$). As in our earlier model, $\gamma_{ki}$ denotes the fraction of its task description that user $k$ offloads in time slot $i$, but in the partial offloading case $(\gamma_{11}+\gamma_{12})$ and $(\gamma_{21}+\gamma_{23})$ lie in the interval $[0, 1]$ rather than at one of the end points. The total energy consumption of a user in the partial offloading case is the summation of the energy consumed in transmitting the offloaded portion to the access point, $E_{\text{off}_k}$, and the energy consumed in the local execution of the remaining fraction, $E_{\text{loc}_k}$. Moreover, the latency constraint of each user must be applied to both the execution time of the local component, $t_{{\text{loc}}_k}$, and the total time that it takes to transmit the offloaded component, execute it at the access point, and send the results back to the user.

For a given choice of offloading fractions, $E_{\text{off}_k}$ takes the same form as in the indivisible case, and the time taken to upload, compute, and return the results of the offloaded portion has the same three components as in \eqref{latency_const}. The uploading time takes a  form $t_{{\text{UL}}_k} = T_s \sum_i \tfrac{\gamma_{ki} B_k }{R_{ki}}$, and we will assume that the time taken to return the results to the user, $t_{{\text{DL}}_k}$, is independent of the fraction of the task that is offloaded. For the class of problems that we are considering, the number of operations to be performed depends on the description length, and hence the execution time at the access point (which has plentiful computational resources) can be modeled as $t_{\text{exe}_k} = \delta_c \sum_i \gamma_{ki} B_k$, where $\delta_c$ denotes the constant processing time of one bit. 

The energy consumed in computing a portion of the task locally, and the time incurred in doing so, are dependent on the computational architecture at the user. Hence, in our initial formulation we will represent them generically using a function $\mathcal{F}_k(\cdot)$ of the number of bits in the retained description, and $t_{{\text{loc}}_k}$, respectively. In this manuscript we will solve the optimal offloading problems for the dynamic voltage scaling architecture \cite{zhang2013energy}, which provides energy-optimal local computation; see Section~\ref{energy_local_opt}. A solution to a related problem for a local computational model that resembles the one we have used for the access point was provided in \cite{salmaniapcc2017}.

With the above computation and communication models in place, the problem of minimizing the total energy consumption of a system with partial offloading can be formulated as 
\begin{subequations}
\label{prob_slots_partial_main}
\begin{align} \notag
\min_{\substack {P_{11}, P_{21}, P_{23}, \\ R_{11}, R_{21}, R_{23}, \\ \gamma_{11}, \gamma_{21}, \gamma_{23}}} \quad & \displaystyle{ \bigl(\tfrac{\gamma_{11} B_1}{R_{11}}\bigr) P_{11}} + \displaystyle{ \bigl(\tfrac{\gamma_{21} B_2}{R_{21}}\bigr) P_{21}}+\displaystyle{\bigl(\tfrac{\gamma_{23} B_2}{R_{23}}\bigr) P_{23}} \\ 
&\quad \quad + \displaystyle{\mathcal{F}_1 \big( (1-\gamma_{11}) B_1\big)} + \displaystyle{\mathcal{F}_2 \big( (1-\gamma_{21} - \gamma_{23}) B_2\big)}\\ 
\label{prob_slots_partial_main_b} 
\text{s.t.}  \enspace \quad \quad \quad & 0 \le \gamma_{11} \le 1,\\
\label{prob_slots_partial_main_c}
\quad & 0 \le \gamma_{21} \le 1,\quad  0 \le \gamma_{23} \le 1,\\
\label{prob_slots_partial_main_d}
 \quad &0 \le \gamma_{21}+\gamma_{23} \le 1,\\
 \label{prob_slots_partial_main_e}
\quad  & T_s (\tfrac{\gamma_{11} B_1}{R_{11}}) + \delta_c \gamma_{11}B_1 \le \bar{L}_1,\\
 \label{prob_slots_partial_main_f}
\quad  & T_s (\tfrac{\gamma_{21}}{R_{21}}+\tfrac{\gamma_{23}}{R_{23}}) B_2+ \delta_c (\gamma_{21}+ \gamma_{23})B_2 \le \bar{L}_2,\\
\label{prob_slots_partial_main_g}
\quad  & t_{{\text{loc}}_k} \le  \bar{L}_k, \quad k=1,2, \\
\label{prob_slots_partial_main_h}
& 0 \le P_{k1}, P_{k2} \le \bar{P}_k, \quad k=1,2,\\
\label{prob_slots_partial_main_i}
\quad & 0 \le R_{23} \le \log_2 (1+ \alpha_2 P_{23}),\\
\label{prob_slots_partial_main_k}
\quad & \{R_{11}, R_{21}\} \in \mathcal{R},
\end{align}
\end{subequations}
where $\bar{L}_k = L_k - t_{\text{DL}_k}$. 

\subsection {Energy-Optimal Local Execution}
\label{energy_local_opt}
In this manuscript, we will consider the dynamic voltage scaling approach to local computation \cite{zhang2013energy, wang2016mobile}. This approach involves adjusting the CPU cycle frequency of the mobile devices so as to minimize the energy required to complete a task within a given deadline. Indeed, for the class of problems that we are considering, the minimum energy required for local processing of $\mu_k B_k$ bits  with a latency constraint of $L_k$ is \cite{zhang2013energy}
\begin{equation}
\label{loc_opt_energy}
E_{\text{loc}_k} = \tfrac{M_k (\mu_k B_k)^3}{{L}_k^2},
\end{equation}
where $M_k$ is a constant that depends on the chip architecture. This expression not only gives us the form of $\mathcal{F}_k(\cdot)$ in \eqref{prob_slots_partial_main}, it also ensures that the local component of the task is completed before the deadline. Therefore, we can remove the local computational latency constraints in \eqref{prob_slots_partial_main_g}.

\subsection{Full Multiple Access Scheme}
\label{partial_mac}
Using the same insights as those used in the complete computation offloading case, we can show that the optimal solution for the problem in \eqref{prob_slots_partial_main} is obtained when the constraints in \eqref{prob_slots_partial_main_e} and \eqref{prob_slots_partial_main_f} hold with equality. Therefore, we can find closed-form expressions for the optimal solutions for $\gamma_{ki}$ in terms of the other parameters of the problem, 
\begin{subequations}
\label{gamm_equs}
\begin{align}
\label{gamm_equs_11}
&\gamma_{11} = \tfrac{\bar{L}_1 R_{11}}{B_1 (T_s+\delta_c R_{11})},\\
\label{gamm_equs_21}
&\gamma_{21} = \tfrac{\bar{L}_1 R_{21}}{B_2 (T_s+\delta_c R_{11})},\\
\label{gamm_equs_23}
&\gamma_{23} = \tfrac{R_{23}}{B_2 (T_s+\delta_c R_{23})} \big(\bar{L}_2 - \tfrac{T_s+\delta_c R_{21}}{T_s + \delta_c R_{11}} \bar{L}_1 \big),
\end{align}
\end{subequations}
where \eqref{gamm_equs_21} results from the fact that $\tau_1= \tfrac{\gamma_{11} B_1}{R_{11}} = \tfrac{\gamma_{21} B_2}{R_{21}}$.

By substituting these closed-form expressions for $\gamma_{ki}$ into \eqref{prob_slots_partial_main}, the remaining optimization problem can be written as
\begin{subequations}
\label{prob_slots_partial_no_gamma}
\begin{align} \notag
\label{prob_slots_partial_no_gamma_a} 
\min_{\substack {P_{11}, P_{21}, P_{23}, \\ R_{11}, R_{21}, R_{23}}} \quad &    \displaystyle{ \tfrac{ \bar{L}_1 }{T_s+\delta_c R_{11}} } (P_{11} + P_{21})+\tfrac{1}{T_s+\delta_c R_{23}}\displaystyle{\big(\bar{L}_2 - \tfrac{T_s+\delta_c R_{21}}{T_s + \delta_c R_{11}} \bar{L}_1 \big) P_{23}} \\ \notag
&~ +\tfrac{M_1}{{L}_1^2}  \big( B_1-\tfrac{\bar{L}_1 R_{11}}{T_s+\delta_c R_{11}} \big)^3\\ 
& ~~ +\tfrac{M_2}{{L}_2^2} \Big( B_2- \tfrac{\bar{L}_2 R_{23}(T_s+\delta_c R_{11}) + \bar{L}_1 T_s (R_{21}-R_{23})}{(T_s+\delta_c R_{11})(T_s+\delta_c R_{23})} \Big)^3\\ 
\label{prob_slots_partial_no_gamma_b} 
\text{s.t.}  \enspace \quad \quad \quad & \eqref{prob_slots_partial_main_b}-\eqref{prob_slots_partial_main_d}, \eqref{prob_slots_partial_main_h}-\eqref{prob_slots_partial_main_i},\\
\label{prob_slots_partial_no_gamma_c}
\quad & 0 \le R_{k1} \le \log_2 (1+ \alpha_k P_{k1}), \quad k=1, 2,\\
\label{prob_slots_partial_no_gamma_d}
\quad & R_{11}+ R_{21} \le \log_2 (1+ \alpha_1 P_{11}+ \alpha_2 P_{21}).
\end{align}
\end{subequations}

In the next step toward solving the problem we obtain closed-form expressions for the transmission powers by decomposing the problem in \eqref{prob_slots_partial_no_gamma} as
\begin{eqnarray}
\label{decomposed_prob_slots_partial_MAC}
&\displaystyle{\min_{ R_{11}, R_{21}, R_{23}}}  \hspace{65pt} \displaystyle{\min_{P_{11}, P_{21}, P_{23}}}    &\eqref{prob_slots_partial_no_gamma_a} \\ 
&\quad \text{s.t.}   \hspace{5pt} \eqref{prob_slots_partial_main_b}-\eqref{prob_slots_partial_main_d},  \hspace{25pt} \text{s.t.} \hspace{5pt}   &\eqref{prob_slots_partial_main_h}, \eqref{prob_slots_partial_main_i}, \eqref{prob_slots_partial_no_gamma_c}-\eqref{prob_slots_partial_no_gamma_d}. \nonumber
\end{eqnarray} 
Given a set of transmission rates $(R_{11}, R_{21}, R_{23})$, the optimal solution for $P_{23}$ is the minimum feasible value, i.e., $P_{23} = \tfrac{2^{R_{23}}-1}{\alpha_2}$ and closed-form expressions for the transmission powers of the users in the first time slot can be obtained by employing the technique that was explained in Section~\ref{general_case}. In this section, we solve the problem for the first subcase of the scenario $\tfrac{\alpha_1}{\alpha_2} \le 1$, which forms an analogy with the first subcase of the complete computation offloading scenario (see Section~\ref{case_A_I_full}). The problem in the other cases can be solved by following similar steps. 

Given the closed-form solutions for all the fractions  $\gamma_{ki}$ and all the transmission powers $P_{ki}$ in terms of the transmission rates, the problem of minimizing the total energy consumption of the users can be reduced to the following three-variable optimization problem 
\begin{subequations}
\label{prob_slots_partial_powers}
\begin{align} \notag
\label{prob_slots_partial_powers_a} 
\min_{R_{11}, R_{21}, R_{23}} \quad &    \displaystyle{ \tfrac{ \bar{L}_1 }{T_s+\delta_c R_{11}} } \bigl(\tfrac{2^{R_{11}}-1}{\alpha_1}+ \tfrac{2^{R_{11}} (2^{R_{21}}-1)}{\alpha_2}\bigr)\\ \notag
&~  +\tfrac{1}{T_s+\delta_c R_{23}}\displaystyle{\big(\bar{L}_2 - \tfrac{T_s+\delta_c R_{21}}{T_s + \delta_c R_{11}} \bar{L}_1 \big) \tfrac{2^{R_{23}}-1}{\alpha_2}} \\ \notag
&~~+\tfrac{M_1}{{L}_1^2}  \big( B_1-\tfrac{\bar{L}_1 R_{11}}{T_s+\delta_c R_{11}} \big)^3\\ 
& ~~~ +\tfrac{M_2}{{L}_2^2} \Big( B_2- \tfrac{\bar{L}_2 R_{23}(T_s+\delta_c R_{11}) + \bar{L}_1 T_s (R_{21}-R_{23})}{(T_s+\delta_c R_{11})(T_s+\delta_c R_{23})} \Big)^3\\ \label{prob_slots_partial_powers_b} 
\text{s.t.}  \enspace \quad \quad \quad & \eqref{prob_slots_partial_main_b}-\eqref{prob_slots_partial_main_d},\\
\label{prob_slots_partial_powers_c1}
\quad & 0 \le R_{23} \le \log_2 (1+ \alpha_2 \bar{P}_{2}),\\
\label{prob_slots_partial_powers_c}
\quad & 0 \le R_{11} \le \log_2 (1+ \alpha_1 \bar{P}_{1}),\\
\label{prob_slots_partial_powers_c2}
\quad & 0 \le R_{21} \le \log_2 (1+ \tfrac{\alpha_2 \bar{P}_{2}}{2^{R_{11}}}).
\end{align}
\end{subequations}
It is shown in Appendix~\ref{App_F_Quasi_convexity} that the objective function of the problem in \eqref{prob_slots_partial_powers} is a quasi-convex function of each of the variables when the other two variables are given. Therefore, the coordinate descent algorithm can be applied to find a stationary solution for the transmission rates \cite[Theorem 1]{hong2016unified}. In all our numerical experiments that approach converged to the globally optimal solution.

\subsection{Time Division Multiple Access Scheme}
\label{partial_TDMA}
For a two-user offloading system that employs the TDMA scheme, each user operates in its own time slot. By simplifying the notation in a natural way,  the total energy minimization problem can be written as
\begin{subequations}
\label{prob_slots_partial_TDMA}
\begin{align}
\label{prob_slots_partial_TDMA_b}
\min_{\substack {P_{1}, P_{2}, R_{1}, R_{2}, \\ \gamma_{1}, \gamma_{2}}} \quad &  \displaystyle{ \bigl(\tfrac{\gamma_{1} B_1}{R_{1}}\bigr) P_{1}} + \displaystyle{ \bigl(\tfrac{\gamma_{2} B_2}{R_{2}}\bigr) P_{2}}   +\tfrac{M_1}{{L}_1^2}\big( (1-\gamma_{1}) B_1\big)^3+\tfrac{M_2}{{L}_2^2} \big( (1-\gamma_{2} ) B_2 \big)^3\\ 
\label{prob_slots_partial_TDMA_b} 
\text{s.t.} \quad \quad & 0 \le \gamma_{k} \le 1, \quad k=1, 2,\\
 \label{prob_slots_partial_TDMA_c}
\quad  & T_s (\tfrac{\gamma_{1} B_1}{R_{1}}) + \delta_c \gamma_{1}B_1 \le \bar{L}_1,\\
 \label{prob_slots_partial_TDMA_d}
\quad  & T_s (\tfrac{\gamma_{1} B_1}{R_{1}})+ T_s (\tfrac{\gamma_{2} B_2}{R_{2}})+ \delta_c \gamma_{2}B_2 \le \bar{L}_2,\\
\label{prob_slots_partial_TDMA_e}
& 0 \le P_{k1}, P_{k2} \le \bar{P}_k, \quad k=1,2,\\
\label{prob_slots_partial_TDMA_f}
\quad & 0 \le R_{k} \le \log_2 (1+ \alpha_k P_{k}), \quad k=1, 2.
\end{align}
\end{subequations}
Since only one user is transmitting during each time slot, it can be shown that the optimal transmission power of each user is the minimum feasible value, i.e., $P_k = \tfrac{2^{R_k}-1}{\alpha_k}$. Moreover, it can be shown that for any optimal solution of the problem in \eqref{prob_slots_partial_TDMA}, the constraint in \eqref{prob_slots_partial_TDMA_d} holds with equality, i.e.,
\begin{equation}
\label{gamma_2}
\gamma_2 = \tfrac{\bar{L}_2-T_s({\gamma_1 B1}/{R_1})}{B_2 ({T_s}/{R_2}+\delta_c)}.
\end{equation}
By substituting the obtained closed-form expressions we can rewrite the problem in \eqref{prob_slots_partial_TDMA} as
\begin{subequations}
\label{prob_slots_partial_TDMA_no_power}
\begin{align} \notag
\label{prob_slots_partial_TDMA_no_power_a}
\min_{R_{1}, R_{2}, \gamma_{1}} \quad &  \displaystyle{ \bigl(\tfrac{\gamma_{1} B_1}{R_{1}}\bigr) \bigl(\tfrac{2^{R_1}-1}{\alpha_1}\bigr)} +\displaystyle{ \bigl(\tfrac{\bar{L}_2 - T_s (\gamma_1 B_1/R_1)}{T_s +R_2 \delta_c}\bigr) \bigl(\tfrac{2^{R_2}-1}{\alpha_2}\bigr)}  \\ 
&~  +\tfrac{M_1}{{L}_1^2}\big( (1-\gamma_{1}) B_1\big)^3+\tfrac{M_2}{{L}_2^2} \big( B_2-\tfrac{\bar{L}_2-T_s({\gamma_1 B1}/{R_1})}{{T_s}/{R_2}+\delta_c}\big)^3 \\
\label{prob_slots_partial_TDMA_no_power_b} 
\text{s.t.}   \quad \quad & 0 \le \gamma_{1} \le 1,\\
\label{prob_slots_partial_TDMA_no_power_c}
\quad & 0 \le \tfrac{\bar{L}_2-T_s({\gamma_1 B1}/{R_1})}{B_2 ({T_s}/{R_2}+\delta_c)} \le 1,\\
\label{prob_slots_partial_TDMA_no_power_d}
\quad & 0 \le R_{k} \le \log_2 (1+ \alpha_k \bar{P}_k), \quad k=1, 2.
\end{align}
\end{subequations}
The three-variable optimization problem in \eqref{prob_slots_partial_TDMA_no_power} is convex in terms of each of the variables when the other two variables are given. Hence, by applying coordinate descent optimization methods a stationary solution of the problem can be obtained. In all of our numerical experiments the coordinate descent algorithm converged to the globally optimal solution. 

In the case of complete computation offloading, we were able to show that when the channel gains of both users are equal and the power budgets are above a threshold, the optimized TDMA scheme obtains the optimal energy consumption of the full multiple access scheme; see Proposition~\ref{prob_2}. As we will formalize in the following proposition, we can extend that result to the case of partial offloading. 
\begin{prop}
\label{prob_2_partial}
Let $\gamma_{ki}^{\star}$ and $P_{ki}^{\star}$ denote an optimal solution for $\gamma_{ki}$ and $P_{ki}$ in \eqref{prob_slots_partial_powers} when $|h_1|^2 = |h_2|^2$. If $P_{11}^{\star}+P_{21}^{\star} \le \min \{ \bar{P}_1, \bar{P}_2 \}$, then the TDMA scheme can obtain the optimal energy consumption of the full multiple access scheme with the offloaded portions of the first and second users equal to $\gamma_1=\gamma_{11}^{\star}$ and $\gamma_2=\gamma_{21}^{\star}+ \gamma_{23}^{\star} $, respectively.  
\end{prop}
\begin{proof}
Let $\hat{\gamma}_{ki}$ denote offloading fractions of an arbitrary instance of the full multiple access scheme. If we select $\gamma_1=\hat{\gamma}_{11}$ and $\gamma_2=\hat{\gamma}_{21}+ \hat{\gamma}_{23}$, then we can apply Proposition~\ref{prob_2} to show that the optimized TDMA scheme achieves the same energy consumption as the full multiple access scheme with offloading fractions $\hat{\gamma}_{ki}$. The proposition follows by looking at the case where the offloading fractions of the full multiple access scheme are optimal; i.e., $\hat{\gamma}_{ki} = \gamma_{ki}^{\star}$.
\end{proof}

\subsection{\seprev{Mixed Binary-Partial Offloading}}
\changen{The problem of minimizing the total energy consumption of a system in which one user has an indivisible computational task while the other user has a divisible task can be treated as a special case of the two-user partial offloading problem in \eqref{prob_slots_partial_main}, in which the fraction of offloaded bits for the user with the indivisible task is either zero or one. In particular, if the first user (the user with the shorter latency) has the indivisible task, then the inequality constraints in \eqref{prob_slots_partial_main_b} will change to $\gamma_{11} \in \{0,1\}$. In the case in which the second user has the indivisible computational task, the inequalities in \eqref{prob_slots_partial_main_d} will change to $ \gamma_{21} + \gamma_{23} \in\{0, 1\}$. These restrictions on the values of $\gamma_{ki}$ for the binary offloading user enable us to obtain simplified expressions for the energy minimization parameters. In particular, when that user does not offload we have a single-user partial offloading problem that has closed-form optimal solution; see Appendix~\ref{single_user_partial}. When the user with the indivisible task is offloading, we can use our decomposition approach to obtain a quasi-closed-form expression that is based on the solution of a two-dimensional optimization problem that is quasi-convex in each variable when the other is fixed; see Appendix~\ref{mixed_off_both}. That is, the fact that we know that one user is completely offloading its task enables us to reduce the dimensions of the optimization problem from the three that were required when both users are partially offloading; (cf. \eqref{prob_slots_partial_powers} and \eqref{prob_slots_partial_TDMA_no_power}).

Furthermore, Proposition~\ref{prob_2_partial} can be extended to the mixed offloading system using a similar proof technique explained for the partial offloading case. That is, when the channel gains are equal, the optimal TDMA scheme achieves the same performance as the optimal FullMA scheme. We will illustrate that result in our numerical results; see Fig.~\ref{main_energy_vs_dist}.}

\section{Numerical Results}
\label{simulation}
In this section we will illustrate the performance of the multiple access computation offloading schemes that we have considered in some simple proof-of-concept experiments that highlight the insights that have been developed. We consider a two-user communication system in which the users have the opportunity to offload their latency-constrained computational tasks to a computationally-rich access point. In Section~\ref{simulation_full} we will illustrate the impact of the choice of the multiple access scheme in the case of complete computation offloading. Then, in Section~\ref{simulation_partial} we will compare the performance of binary and partial computation offloading under the full multiple access and TDMA schemes.

\subsection{Complete Computation Offloading}
\label{simulation_full}
In our first experiment, we examine the total energy usage of two offloading users as the \rev{power} gain of user 1\rev{'s channel}, $|h_1|^2$, changes. The symbol period of the channel is set to be $T_s=10^{-6} \text{s}$, and we set the power budgets, the latencies, the channel gain of the second user and the receiver noise variance to be constant values, namely, $\bar{P}_1= 0.3 T_s$, $\bar{P}_2= 0.5 T_s$, $L_1 = 2.5$s, $L_2=3.3$s, $|h_2|^2=0.1$, and $\sigma^2 = 0.1T_s$, respectively. The number of bits that are needed to describe the tasks to be offloaded are $B_1=B_2=10^6$ and we set the sum of the time of execution of the application in the cloud and the time it takes to download the result to the mobile users to be $T_1=T_2 =0.5$s. 
\begin{figure}
\centering
\includegraphics[width=1\linewidth,width=0.48\textwidth]{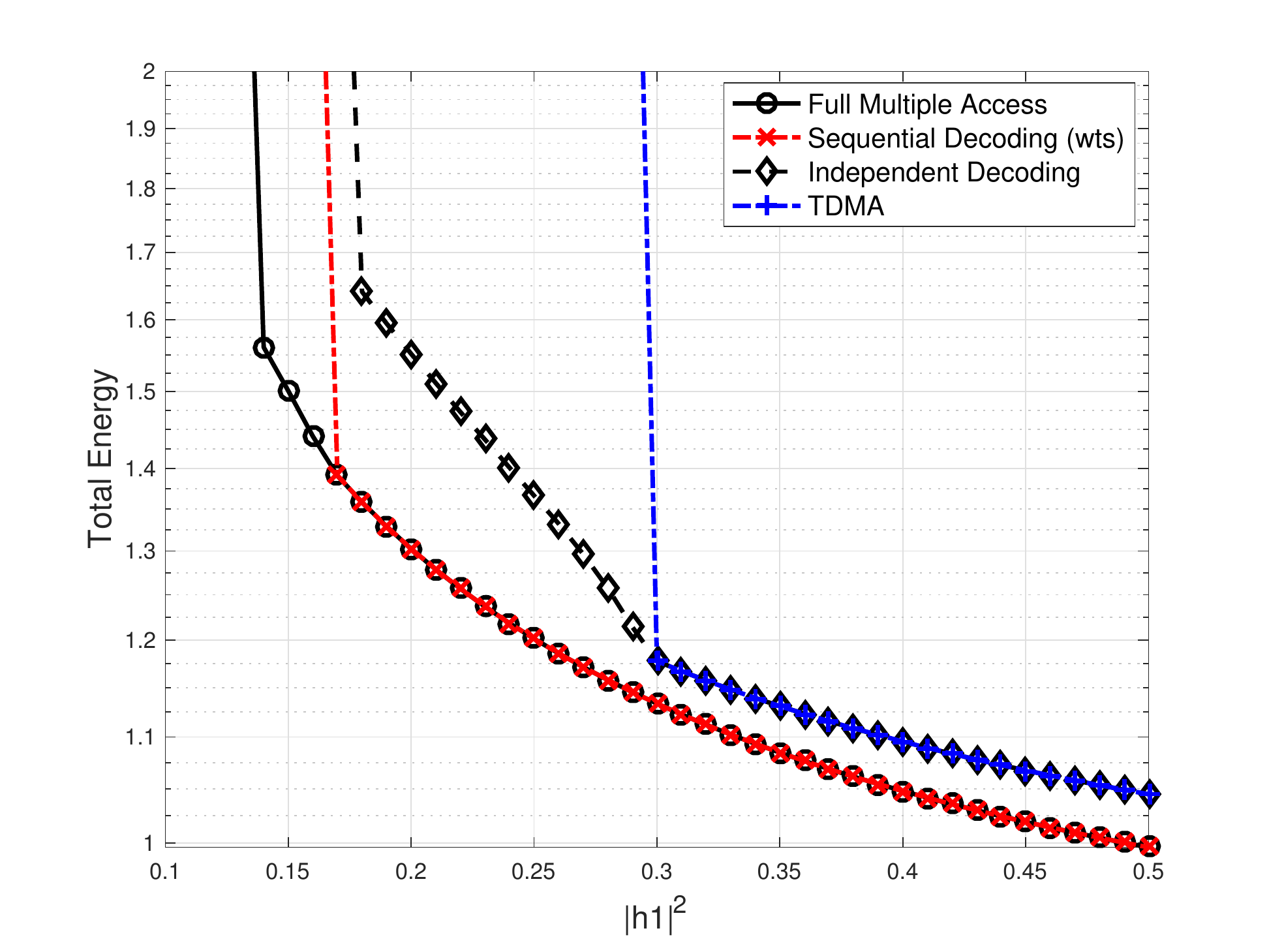} 
\caption{Energy required to offload the tasks for the optimal, sequential decoding without time sharing, TDMA, and independent decoding schemes \rev{as a function of $|h_1|^2$}.} 
\label{main_all_Infeas_TDMA}
\end{figure}

It can be seen from Fig.~\ref{main_all_Infeas_TDMA} that for small values of the \rev{power} gain of user 1 only the full multiple access scheme is able to offload both tasks while meeting the constraints. This implies that the optimized transmission rates of the first time slot are on the dominant face of the capacity region of the multiple access channel, cf. Fig.~\ref{feas_region_all_MAC}. For larger values of $|h_1|^2$, sequential decoding (wts) is also feasible and it can be seen that it has the same total energy consumption as the full multiple access scheme. 
As the value of $|h_1|^2$ increases, the independent decoding scheme becomes feasible, but the total energy consumption of independent decoding is significantly greater than the energy consumed by the full multiple access and sequential decoding (wts) schemes. As $|h_1|^2$ is increased further, the TDMA scheme eventually becomes feasible. As argued in Section~\ref{Ind_Dec}, once it becomes feasible, it achieves the same energy consumption as the independent decoding scheme. (In TDMA, the decoders work independently.) However, for this range of values of $|h_1|^2$ the full multiple access and sequential decoding (wts) schemes are able to offload both tasks using less energy.

In \rev{our next} experiment we will \rev{illustrate} the impact of the latency of the second user's task, $L_2$. To do so, we tighten the first user's latency constraint to $L_1 = 1.8$s and we provide user 2 with a larger channel gain, \rev{$|h_2|^2=0.24$}. The number of bits needed to describe the tasks are changed to $B_1=3 \times 10^6$ and $B_2=5 \times 10^6$, \rev{and the receiver noise variance is set to be $\sigma^2 = 2\times 10^{-3}T_s$.} The other system parameters remain the same. 

In Fig.~\ref{main_fig} we plot the energy required to offload both tasks as a function of the first user's channel gain, $|h_1|^2$, for two values of the latency for the second user, namely $L_2=2$s and $L_2= 2.6$s.
\begin{figure}
\centering
\includegraphics[width=1\linewidth,width=0.48\textwidth]{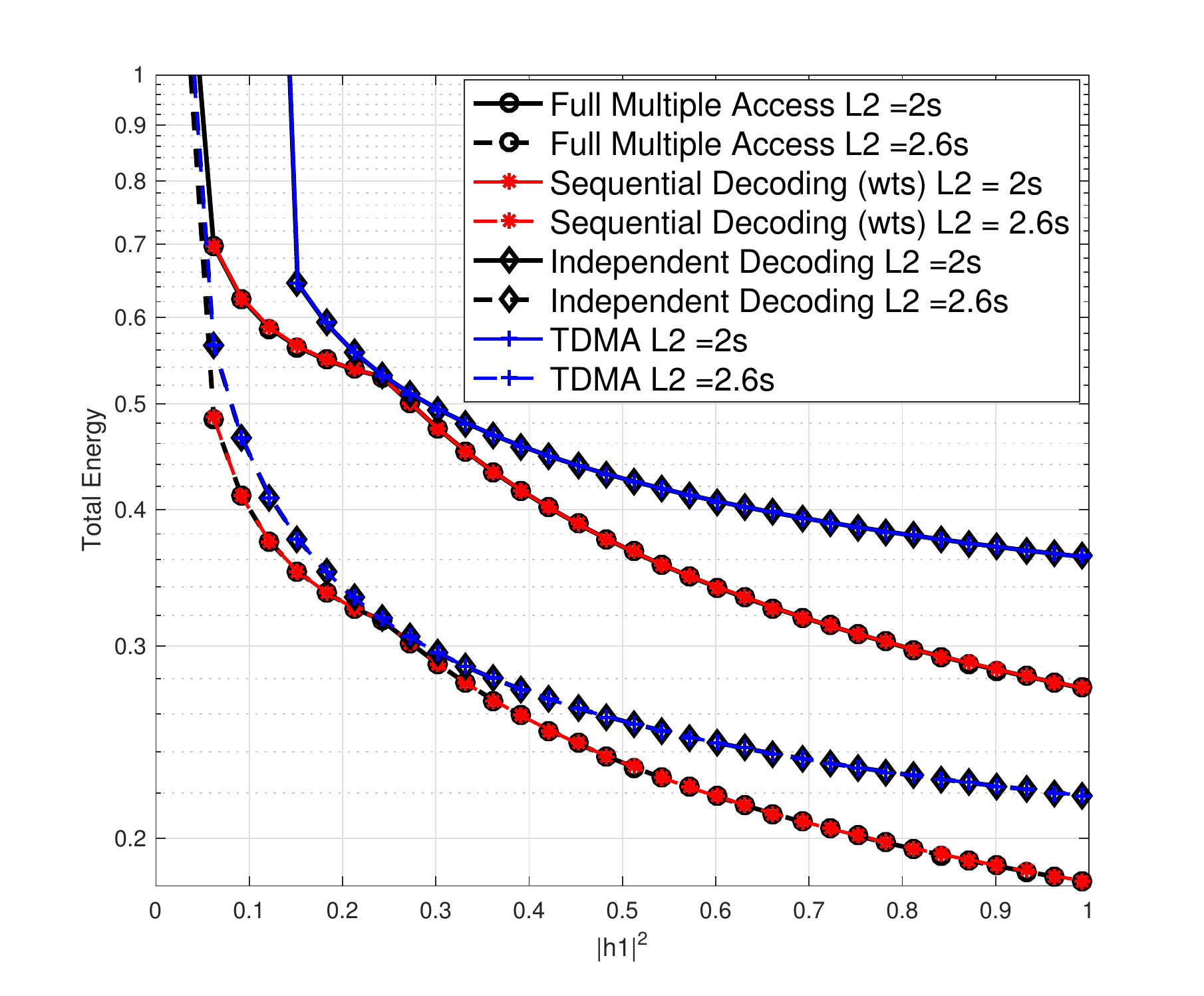} 
\caption{Energy required to offload the tasks \rev{as a function of $|h_1|^2$} for two different values for the second user's latency, $L_2=2$s and $L_2= 2.6$s. } 
\label{main_fig}
\end{figure}

\begin{figure}
\centering
\includegraphics[width=1\linewidth,width=0.48\textwidth]{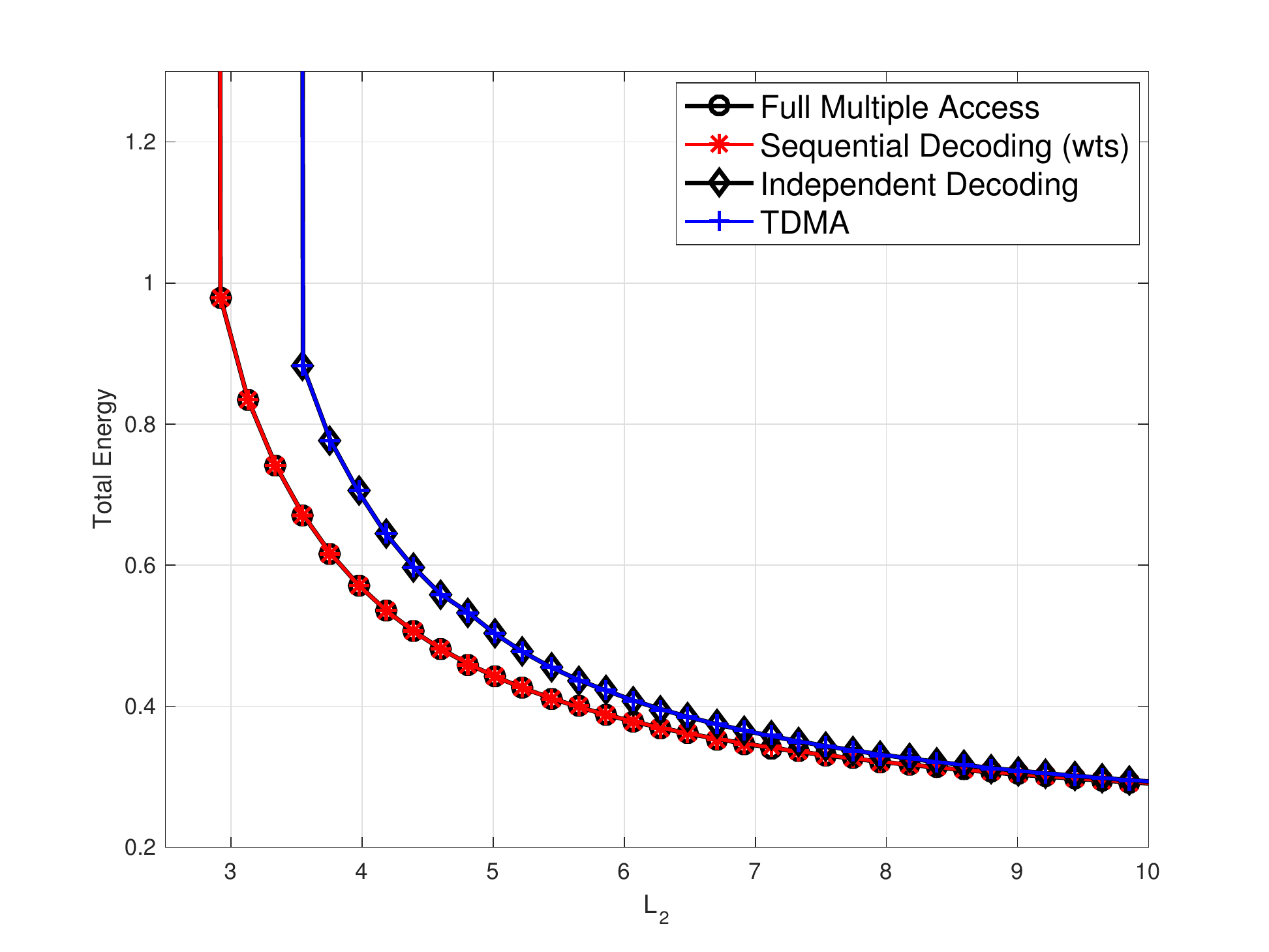} 
\caption{Energy required to offload the tasks \rev{as a function of $L_2$, the latency of the second user's application.}} 
\label{main_fif_2}
\end{figure}
The sets of curves for the two latencies exhibit similar characteristics, but these characteristics, and the reduced energy consumption of the full multiple access scheme, are more pronounced in the case where the latency is tighter. (As expected more energy is required to offload the tasks in that case.) Both sets of curves in Fig.~\ref{main_fig} demonstrate the ability of the full multiple access and sequential decoding (wts) schemes to take advantage of skewed channel conditions. In particular, when $|h_1|^2$ is small these schemes are able to offload both tasks, whereas TDMA and independent decoding are unable to do so. (The extended range of the full multiple access and sequential decoding (wts) schemes is quite significant in the lower latency case.) When $|h_1|^2$ is large, the full multiple access and sequential decoding (wts) schemes are able to provide a substantial reduction in the energy required to perform the offloading. Fig.~\ref{main_fig} also illustrates the impact of Proposition~\ref{prob_2}; namely that when the channel gains are equal and the power budgets are above an explicit threshold, TDMA, independent decoding and sequential decoding (wts) are all able to achieve the same minimum energy consumption as the full multiple access scheme.  

In Fig.~\ref{main_fif_2} we plot the energy required to offload both tasks as a function of the second user's latency constraint, $L_2$, for the case in which $|h_1| ^2 = 0.6$ and $|h_2| ^2 = 0.06$. In order to guarantee the feasibility of all four multiple access schemes we set the latency of user 1 equal to $L_1 = 2$s and the transmitted number of bits for the first and the second users are $B_1 = 4\times 10^6$ bits  and $B_2 = 8\times 10^6$ bits, respectively. The other parameters are the same as those that were used for Fig.~\ref{main_fig}. 

Fig.~\ref{main_fif_2} reinforces the observation from Fig.~\ref{main_fig} that as the latency constraints are tightened, the ability of the full multiple access scheme to exploit all the capabilities of the multiple access channel offers increasing performance gains. Fig.~\ref{main_fif_2} also illustrates the fact that the full multiple access scheme can satisfy tighter latency constraints than the TDMA and independent decoding schemes. (For all the values of $L_2$ that we considered in Fig.~\ref{main_fif_2}, sequential decoding (wts) is optimal and TDMA is an optimal scheme for independent decoding.)

\subsection{Binary and Partial Computation Offloading}
\label{simulation_partial}
In the second phase of our numerical analysis we study the case in which the computational tasks of the users are infinitesimally divisible and hence partial offloading can be employed. In this phase, we consider the full multiple access and TDMA schemes, and we examine the total energy consumption of partial computation offloading under different parameter settings. We also compare the energy consumption of partial offloading to that of the binary offloading scheme that would be used if the tasks were considered to be indivisible. 

We first illustrate the performance of the full multiple access and TDMA schemes as a function of the channel gain of user 1, $|h_1|^2$. We set the power budgets, the latencies, the channel gain of the second user and the receiver noise variance to be constant values, namely, $\bar{P}_1= \bar{P}_2= 0.5T_s$, $L_1 = 1.5$s, $L_2=2$s, $|h_2|^2=0.5$, and $\sigma^2 = 10^{-3}T_s$, respectively. The number of bits to describe the problems are $B_1=2 \times 10^6$ and $B_2=6 \times 10^6$ and we set the time it takes to download the result to the mobile users to be $t_{\text{DL}_1}=t_{\text{DL}_2} =0.2$s. As explained in Section~\ref{partial_offloading}, we consider data-partitionable computational tasks for which the (optimal) local energy consumption can be modeled as a function of number of bits; see Section~\ref{energy_local_opt}. In order to be consistent with the measurements in \cite{miettinen2010energy}, we set the constants $M_k$ in the local energy consumption expression in \eqref{loc_opt_energy} to $M_1 = M_2 = 10^{-18}$ \cite{zhang2013energy, wang2016mobile}.

Fig.~\ref{partial_sweeping_h_stochastic_energy} illustrates the total energy consumption of the users in the partial and binary computation offloading scenarios for the full multiple access and TDMA schemes, and Fig.~\ref{partial_sweeping_h_stochastic_gammas} illustrates the corresponding fraction of the bits that each user offloads to the access point in the partial offloading scenario. It can be seen from Fig.~\ref{partial_sweeping_h_stochastic_energy} that in both the partial and binary offloading scenarios taking advantage of the full capabilities of the channel enables the users to execute their tasks with substantially less energy consumption compared to the TDMA scheme and the gap between the energy usage of these schemes becomes larger as the channel gain of the first user increases. Fig.~\ref{partial_sweeping_h_stochastic_energy} also illustrates the fact that, since the power budgets are above the threshold in Proposition~\ref{prob_2_partial}, when the channel gains are equal TDMA can achieve the minimum energy consumption.
Another observation in Fig.~\ref{partial_sweeping_h_stochastic_energy} is that for large values of $|h_1|^2$, binary offloading with the full multiple access scheme achieves lower energy consumption than partial offloading using TDMA. This is due to the fact that the full multiple access scheme's ability to utilize all the capabilities of the channel overcomes the limitations of binary offloading when the channel gains are sufficiently different. 

It can be seen in Fig.~\ref{partial_sweeping_h_stochastic_gammas} that using the full capabilities of the channel enables the users to offload larger fraction of bits to the access point than TDMA. This results in a lower total energy consumption; see Fig.~\ref{partial_sweeping_h_stochastic_energy}. Moreover, when the channel gains are equal, the portions that the users offload in the TDMA scheme are the optimal portions offloaded by the corresponding users in the full multiple access scheme, which verifies Proposition~\ref{prob_2_partial}. 

Fig.~\ref{partial_sweeping_h_stochastic_gammas} exhibits that in the TDMA case, as one would expect, an increase in the channel gain of the first user leads to an increased fraction of bits that each user offloads. For the full multiple access scheme, an increase in the channel gain of the first user  leads to an increase in the fraction of bits offloaded by that user. This observation can be verified from the expression in \eqref{gamm_equs_11} and the fact that by increasing the channel gain of the first user, a higher transmission rate will be employed by that user. However, the offloaded fraction of the second user does not change in a monotonic manner. When $|h_1|^2 \le |h_2|^2$, an increase in the channel gain of user 1 results in a decrease in the portion of bits offloaded by user 2, while for $|h_1|^2 \ge |h_2|^2$ the offloaded portion of user 2 increases. That is because the minimum energy consumption of the system depends on the ratio between the channel gains (see Section~\ref{general_case}), which also affects the portion of  bits offloaded by the second user; see  \eqref{gamm_equs_21} and \eqref{gamm_equs_23}. 
\begin{figure}
\centering
\includegraphics[width=1\linewidth,width=0.48\textwidth]{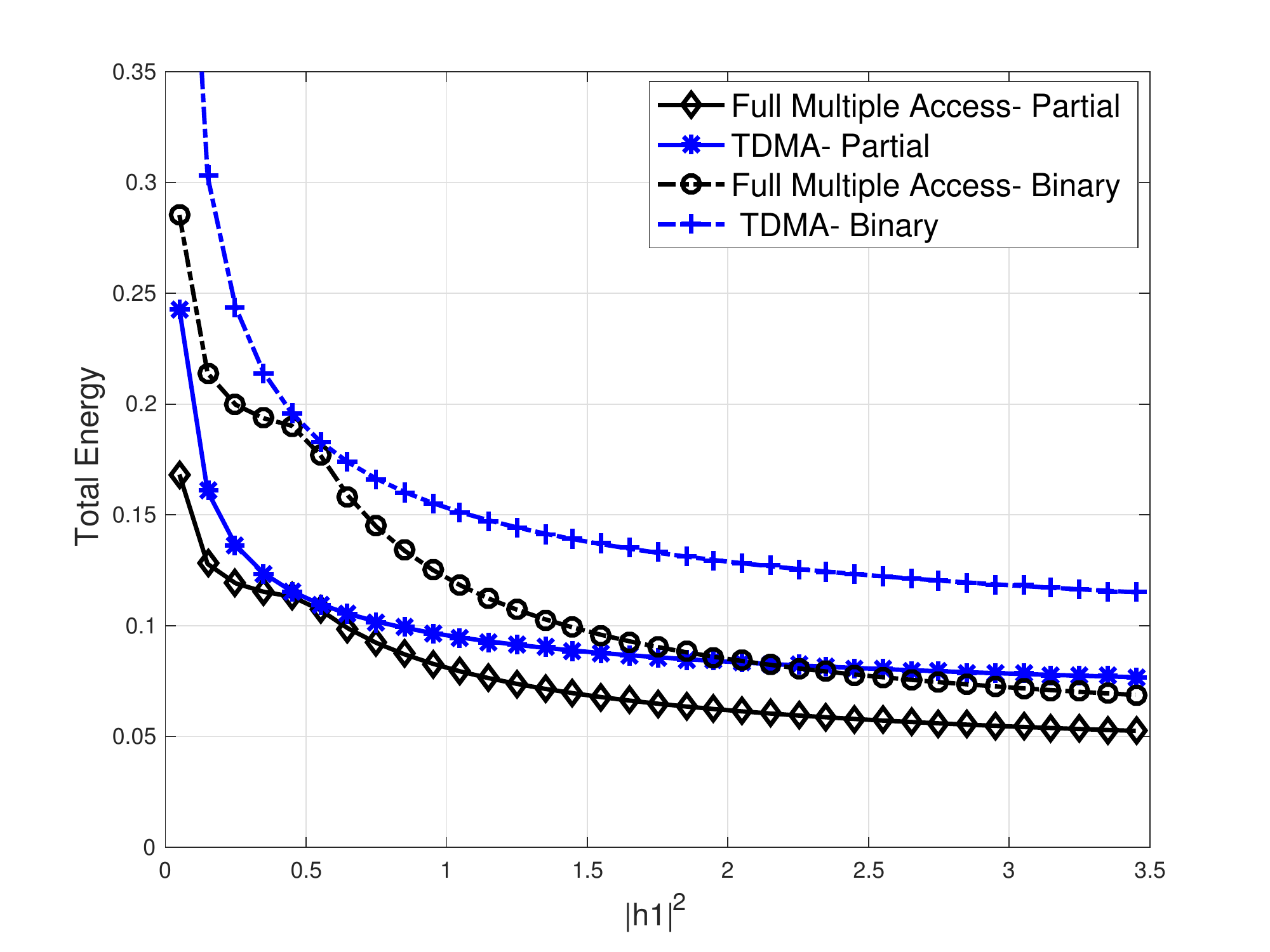} 
\caption{The total energy consumption of the two-user system with the full multiple access and TDMA schemes in the cases of binary and partial computation offloading \rev{as a function of $|h_1|^2$.}} 
\label{partial_sweeping_h_stochastic_energy}
\end{figure}

\begin{figure}
\centering
\includegraphics[width=1\linewidth,width=0.48\textwidth]{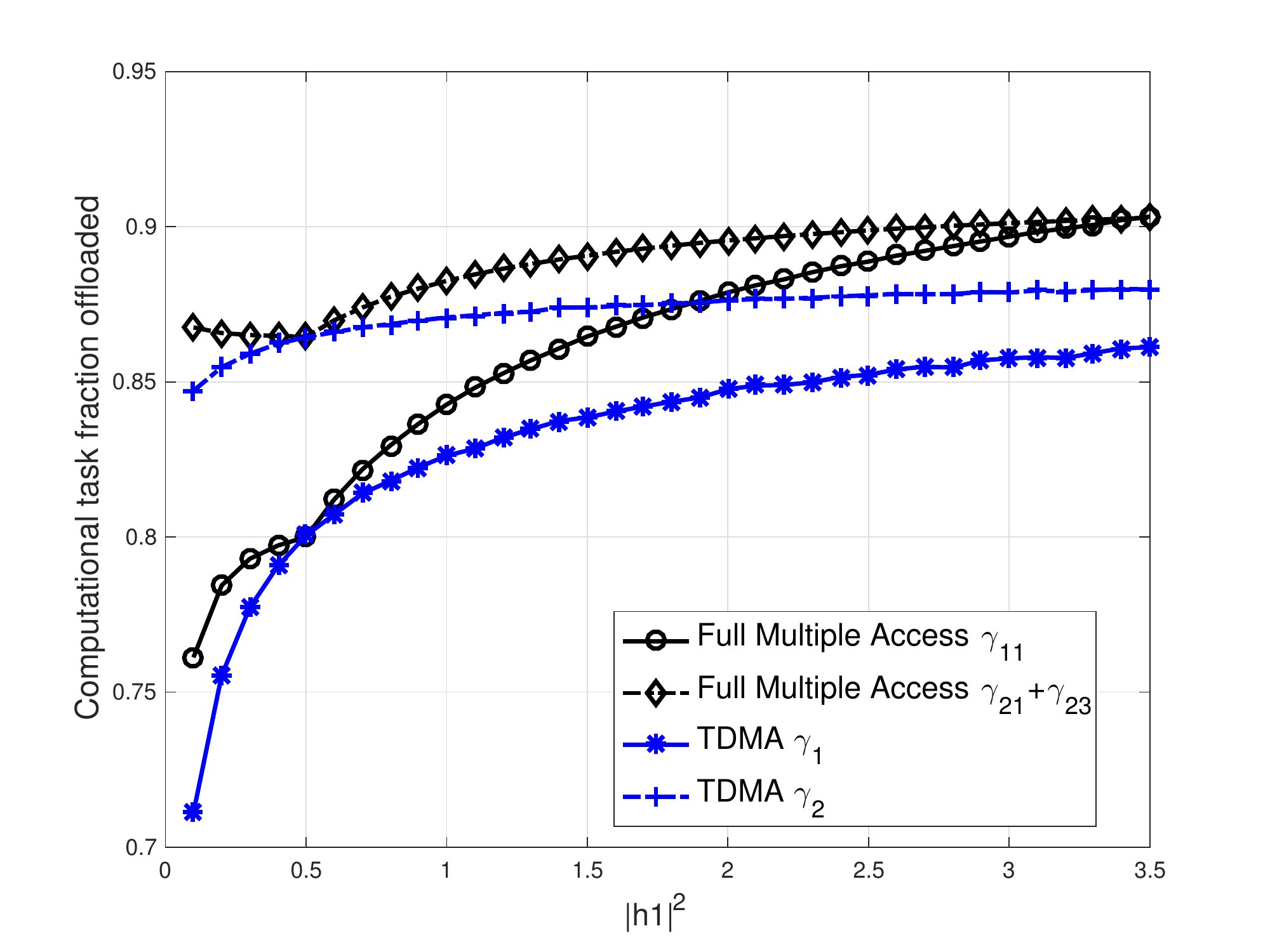} 
\caption{The fraction of the total number of bits offloaded by each user with the full multiple access and TDMA schemes in the partial computation offloading case \rev{as a function of $|h_1|^2$.}} 
\label{partial_sweeping_h_stochastic_gammas}
\end{figure}

Our final numerical experiments examine the average performance of the full multiple access and TDMA schemes under a simple fading channel model for both binary and partial computation offloading. The channel model has a (large-scale) path-loss exponent of $3$ and the small-scale fading has a Rayleigh distribution, and the variance of Gaussian noise is set to be $\sigma^2 = 10^{-13}$. The latency constraints of the tasks are $L_1 = 1.7$s, $L_2 = 2$s and the descriptions of the tasks require $B_1 = 2 \times 10^6$ and $B_2 = 5\times 10^6$ bits, respectively. The second user is placed at a distance of $500$m from the access point, and in Fig.~\ref{main_energy_vs_dist} we plot the average energy required to offload the tasks as user 1 moves from a position $100$m from the access point to a position $900$m away. (In the mixed offloading case, user 1 makes a binary offloading decision and user 2 partially offloads its task.)
The average is taken over $10^5$ realizations of the channel pairs for which the two schemes provided a feasible solution in the complete and partial computation offloading cases. The other parameters of the problem are set to be the same as those in the experiment that produced Figs~\ref{partial_sweeping_h_stochastic_energy} and \ref{partial_sweeping_h_stochastic_gammas}.

Figs~\ref{main_energy_vs_dist} and \ref{main_fraction_vs_dist} demonstrate that the insights that were developed analytically in Section~\ref{insight} for individual channel realizations also reflect the performance on average. In particular, when the users are at similar distances from the access point, then the channel gains are likely to be similar and hence we would expect the performance of TDMA to be close to that of the full multiple access scheme. This is indeed the case in Figs~\ref{main_energy_vs_dist} and \ref{main_fraction_vs_dist}. When the users are at significantly different distances from the access point, their channel gains are likely to be quite different, and hence we would expect the full multiple access scheme to have significantly better performance than TDMA. Once again, Fig.~\ref{main_energy_vs_dist} confirms that insight, and Fig.~\ref{main_fraction_vs_dist} shows how the full multiple access scheme enables a greater fraction of each task to be offloaded. Indeed, when user 1 is far from the access point, binary offloading with the full multiple access scheme consumes less energy than partial offloading with TDMA. 
\begin{figure}
\centering
\includegraphics[width=1\linewidth,width=0.45\textwidth]{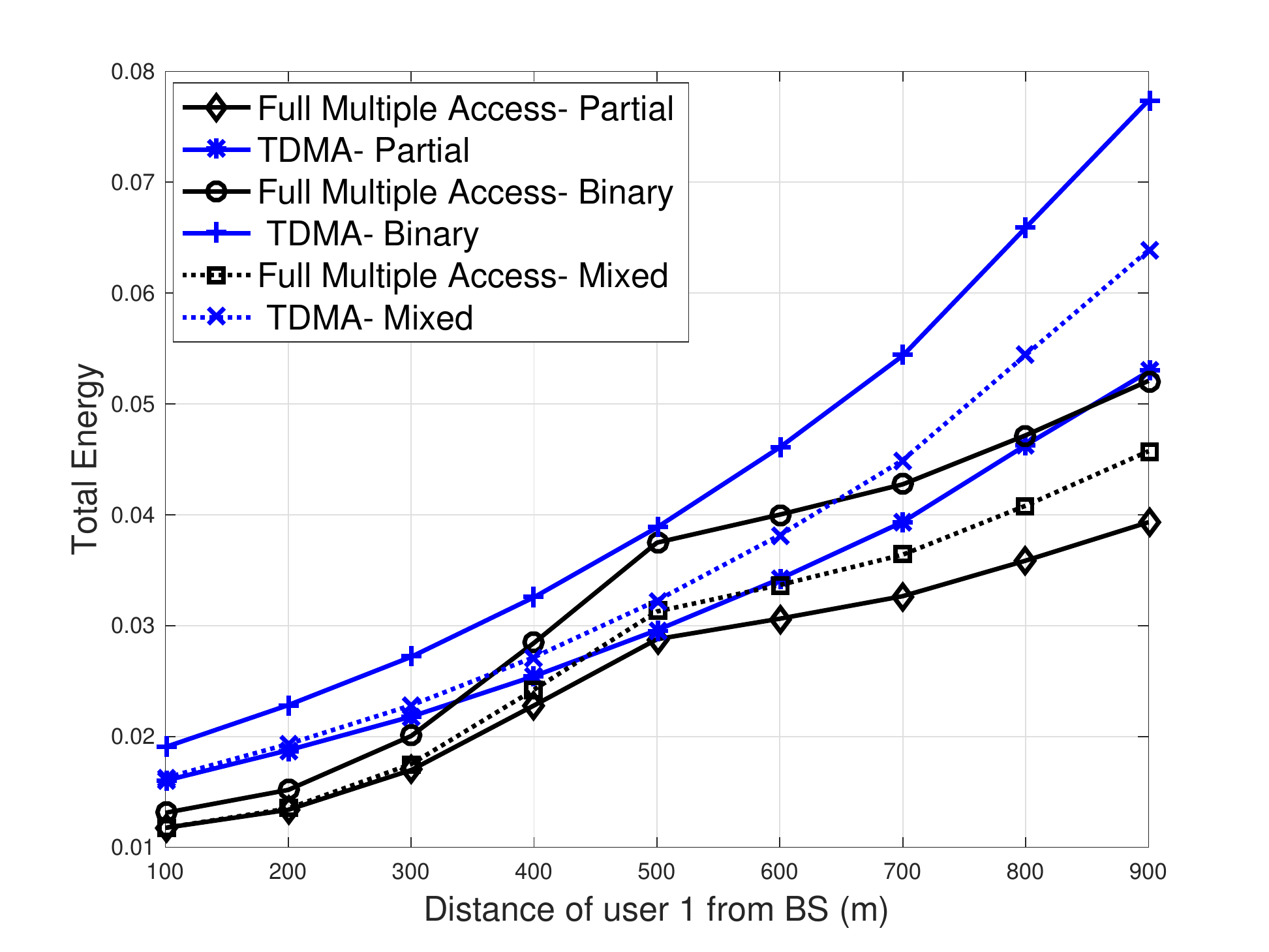}
\caption{Average energy required to offload the computational tasks for the full multiple access and TDMA schemes against the distance of user 1 from the access point in the binary and partial computation offloading scenarios. User 2 is $500$m from the access point.} 
\label{main_energy_vs_dist}
\end{figure}

\begin{figure}
\centering
\includegraphics[width=1\linewidth,width=0.45\textwidth]{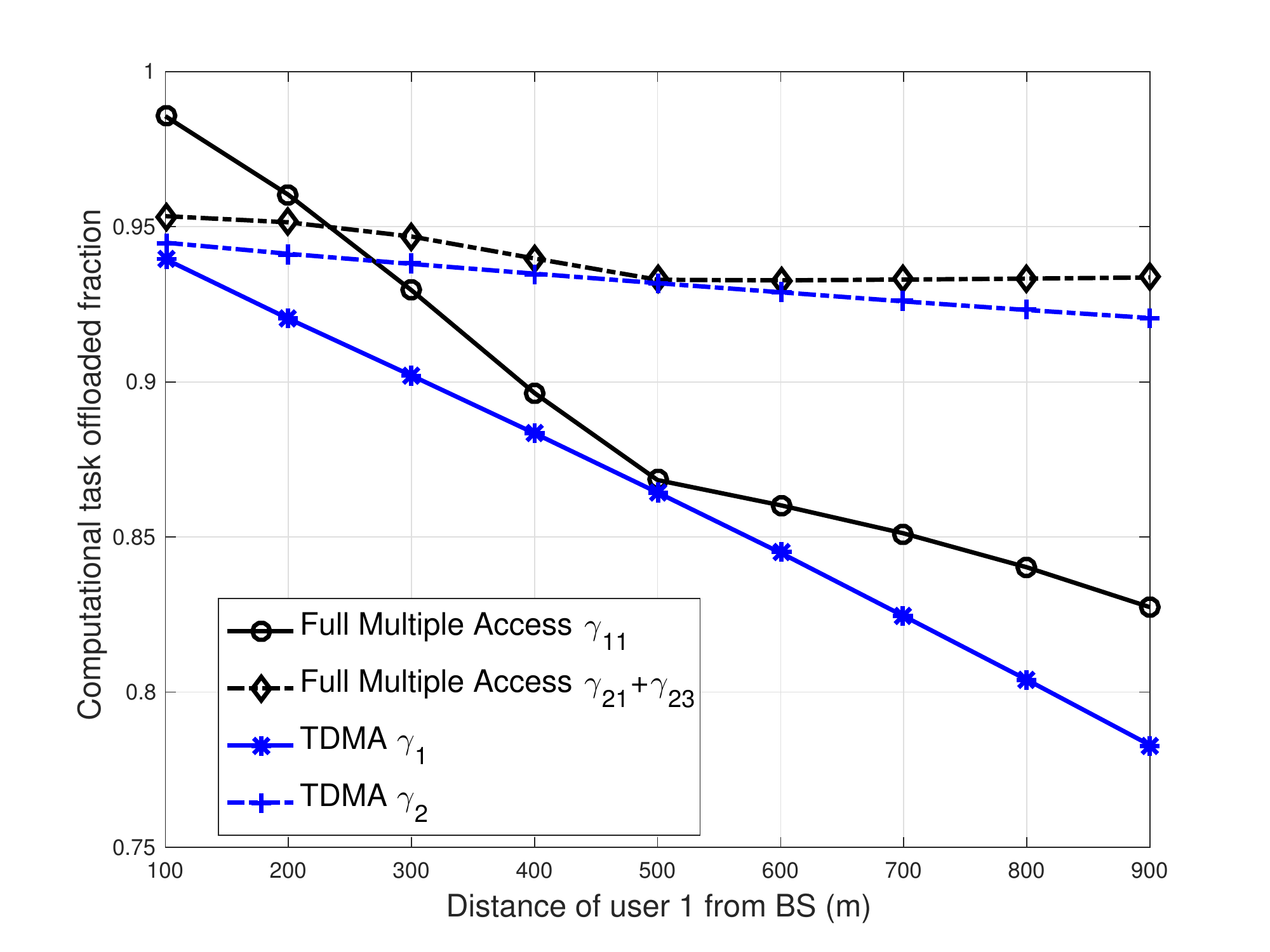} 
\caption{Average computation fraction offloaded by the users for the full multiple access and TDMA schemes against the distance of user 1 from the access point in the partial computation offloading scenario. User 2 is $500$m from the access point.} 
\label{main_fraction_vs_dist}
\end{figure}

\section{Conclusion}
In this work, we have obtained closed-form and quasi-closed-form solutions to problems of optimizing the communication resource allocation so as to minimize the energy that the users expend in a computational offloading system with two users \rev{and plentiful computational resources}. We have provided solutions for both indivisible and infinitesimally divisible computational tasks, and we consider the full multiple access scheme along with some of the simplified schemes, namely, TDMA, sequential decoding (without time sharing), and independent decoding. In broad terms, the structure of our solutions suggests that if the channel gains of the users are similar (or if the latency constraints are loose), then the implementation simplicity of TDMA may outweigh the energy reduction offered by the full multiple access scheme. However, when the latency constraints are tight and the channel gains are quite different from each other, the full multiple access scheme provides a substantial reduction in the energy required to complete the tasks.

\appendices
\section{Objective Function in \eqref{general_inner_1_R} is Increasing}
\label{increaseing_pbj_func}
After removing the positive coefficient $(\tfrac{B_2-\tilde{L}_1 R_{21}}{\alpha_2})$, the first derivative of the objective function in \eqref{general_inner_1_R} with respect to $R_{23}$ is
\begin{align}
\label{dervative_R1_one_user_spl_case}
\tfrac{d }{d R_{23}}  \bigl(\tfrac{ 2^{R_{23}}-1}{R_{23}}\bigr)& =  \bigl[   \tfrac{R_{23} 2^{R_{23}} \ln 2  - (2^{R_{23}}-1)}{R_{23}^2}  \bigr]. 
\end{align}
To show that the numerator on the right hand side of \eqref{dervative_R1_one_user_spl_case} is positive for $R_{23} > 0$, we observe that 
\begin{align}
\label{dervative_R1_one_user_2_spl_case}
\tfrac{d}{d R_{23}}  \bigl(R_{23} 2^{R_{23}} \ln 2 - (2^{R_{23}}-1)\Bigr)& = R_{23} 2^{R_{23}} (\ln 2)^2,
\end{align}
which is positive for $R_{23}>0$. Since the numerator of \eqref{dervative_R1_one_user_spl_case} is zero for $R_{23}=0$, the positivity of \eqref{dervative_R1_one_user_2_spl_case} implies that the expression in \eqref{dervative_R1_one_user_spl_case} is positive for $R_{23}> 0$ and hence that the objective in \eqref{general_inner_1_R} is an increasing function of $R_{23}$ over the feasible set.

\section{Optimality of Two Time Slot Scenario}
\label{Two_Time_Slots}
Let us consider the three-time-slot system illustrated in Fig.~\ref{Time_Slots_Three_Two_a}, with the notation for the time slot durations, the transmission powers and transmission rates being as defined in Section~\ref{sys_model}. Let us also consider the two-time-slot system depicted in Fig.~\ref{Time_Slots_Three_Two_b}, in which the  power and the rate of the $k$\textsuperscript{th} user in the first time slot are denoted by $P'_k$ and $R'_k$, respectively. We will show that for the full multiple access scheme, if the users' tasks can be offloaded using the three-time-slot system with a given energy, then there exists a power and rate allocation for the two-time-slot system that can offload the tasks using the same energy. It will suffice to assume that the power and rate allocations in the second time slot of the two-time-slot system are the same as those in the third time slot of the three-time-slot system. Furthermore, it will suffice to add the constraint that the energy consumption of each user is the same in both systems. 
\begin{figure}
\begin{subfigure}{.45\linewidth}
\centering
\includegraphics[width=0.65\linewidth]{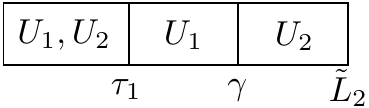}
\caption{Three-time-slot system}
\label{Time_Slots_Three_Two_a}
\end{subfigure}%
\begin{subfigure}{.45\linewidth}
\centering
\includegraphics[width=0.65\linewidth]{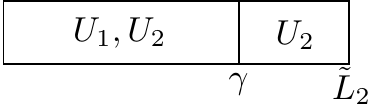}
\caption{Two-time-slot system}
\label{Time_Slots_Three_Two_b}
\end{subfigure}\\%
\caption{In the three-time-slot system there is one time slot (the second) in which user 1 is the only user transmitting. In the two-time-slot system user 1 completes its transmission in the first time slot and only user 2 has a slot in which it transmits alone. }
 \label{Time_Slots_Three_Two}
\end{figure}

For the energy and the number of transmitted bits of each user to be the same, $P_k^\prime$ and $R_k^\prime$ must satisfy 
\begin{subequations}
\label{power_rate_two}
\begin{align} 
& \tau_1P_{11} + (\gamma-\tau_1)  P_{12} = \gamma P'_{1}  \\
& \tau_1 P_{21} + (\tilde{L}_2-\gamma)  P_{22}  = \gamma P'_{2} + (\tilde{L}_2-\gamma)  P_{22}, \\
&  \tau_1 R_{11} +  (\gamma-\tau_1) R_{12}  = \gamma R'_{1} ,\\
& \tau_1 R_{21} + (\tilde{L}_2-\gamma) R_{22} = \gamma R'_{2} + (\tilde{L}_2-\gamma)R_{22}. 
\end{align}
\end{subequations}
The solution of that set of linear equations is
\begin{subequations}
\begin{align}
\label{powers_sec_case}
& P'_{1} = \tfrac{\tau_1}{\gamma}  P_{11} + \tfrac{\gamma-\tau_1}{\gamma} P_{12} \quad  \text{and} \quad   P'_{2} = \tfrac{\tau_1}{\gamma} P_{21},\\ 
\label{rates_sec_case}
& R'_{1} = \tfrac{\tau_1}{T}  R_{11} + \tfrac{\gamma-\tau_1}{ \gamma} R_{12} \quad  \text{and} \quad   R'_{2} = \tfrac{\tau_1}{\gamma} R_{21}.  
\end{align}
\end{subequations}
What remains is to show that these power and rate allocations satisfy the rate region constraint for the first time slot of the two-time-slot system, namely,
\begin{subequations}
\label{rates_constraints_sec_case_two_slots}
\begin{align}
& R'_{1} \le \log_2(1+\alpha_1 P'_{1}  ), \quad R'_{2} \le \log_2(1+\alpha_2 P'_{2}  ),\\
& R'_{1} +  R'_{2} \le \log_2(1+\alpha_1 P'_{1}+ \alpha_2 P'_{2}   ).
\end{align}
\end{subequations}
The inequalities in \eqref{rates_constraints_sec_case_two_slots} can be rewritten in terms of the rates and powers of the three-time-slot case as follows,
\begin{subequations}
\label{rates_constraints_sec_case}
\begin{align}
\label{rates_constraints_sec_case_1}
& \tfrac{\tau_1}{T}  R_{11} + \tfrac{\gamma-\tau_1}{ \gamma} R_{12} \le \log_2 \bigl(1+\alpha_1\rho   \bigr),\\
\label{rates_constraints_sec_case_2}
& \tfrac{\tau_1}{\gamma} R_{21} \le \log_2 \bigl(1+\alpha_2 (\tfrac{\tau_1}{\gamma} P_{21})   \bigr),\\
\label{rates_constraints_sec_case_3}
&\tfrac{\tau_1}{T}  R_{11} + \tfrac{\gamma-\tau_1}{ \gamma} R_{12} +\tfrac{\tau_1}{\gamma} R_{21}  \le \log_2(1+\alpha_1 \rho + \alpha_2 (\tfrac{\tau_1}{\gamma} P_{21})  \bigr),
\end{align}
\end{subequations}
where $\rho = \tfrac{\tau_1}{\gamma}  P_{11} + \tfrac{\gamma-\tau_1}{\gamma} P_{12}$.
To establish the validity of the inequalities in \eqref{rates_constraints_sec_case}, we use the rate constraints of the three-time-slot case and the concavity of the logarithm. For example, since the power and rate allocations for the three-time-slot case are assumed to be valid, we have that $R_{21} \leq \log_2(1 + \alpha_2P_{21})$. Using that and the concavity of the logarithm we have that
\begin{equation}
\label{rates}
\tfrac{\tau_1}{\gamma}  R_{21} \le \tfrac{\tau_1}{\gamma}  \log_2 (1+\alpha_2  P_{21}) \le \log_2 \bigl(1+\alpha_2 (\tfrac{\tau_1}{\gamma}  P_{21})   \bigr),
\end{equation}
and hence that \eqref{rates_constraints_sec_case_2} holds. The inequalities in \eqref{rates_constraints_sec_case_1} and \eqref{rates_constraints_sec_case_3} can be established in an analogous way.

\section{Convexity of Objective Function in \eqref{prob_slots_general_Rs_K_a}}
\label{app_cvx_R21_case_I}
Since the exponential function is convex, the second term of \eqref{prob_slots_general_Rs_K_a} is convex. To prove the convexity of the third term, and hence the convexity of the function as a whole, we let $f_3 (R_{21})$ denote the third term of the objective function and evaluate its second derivative,
$\textstyle{f_3''(R_{21}) = \ln (2)^2 (\tfrac{\tilde{L}_1^2}{\alpha_2 (\tilde{L}_2-\tilde{L}_1)}) ~ 2^{\frac{B_2-\tilde{L}_1R_{21}}{\tilde{L}_2-\tilde{L}_1}}}$.
By our standing assumption that $\textstyle{ \tilde{L}_2 \ge \tilde{L}_1}$, $f_3''(R_{21}) \ge 0 $ and hence $f_3 (R_{21})$ is convex.  

\rev{\section{Optimality of TDMA for Independent Decoding}
\label{TDMA_opt_IndDec}
Consider the structure of TDMA signalling illustrated in Fig.~\ref{TDMA_over_Ind_Dec_a}, where $x$ and $y$ are portions of the duration of the first and second users' transmissions, respectively. In the independent decoding signalling structure illustrated  in Fig.~\ref{TDMA_over_Ind_Dec_b}, the intervals $x$ and $y$ are combined, with both users transmitting simultaneously and being decoded independently. Let $P_k$ and $R_k$ denote the transmission power and transmission rate of the $k^{th}$ user in the TDMA scheme, respectively, and $P'_k$ and $R'_k$ denote those quantities in the first time slot of the independent decoding scheme. Since the number of transmitted bits in both cases must be the same, we have that
\begin{subequations}
\label{TDMA_over_Ind_Energy}
\begin{align}
& (x+y) R'_1 = xR_1= x\log_2(1+\alpha_1P_1) \\
& (x+y) R'_2 = yR_2 = y\log_2(1+\alpha_2P_2).
\end{align}
\end{subequations}
We will show that we cannot find a set of transmission rates and transmission powers for which the energy consumption in independent decoding scheme is less than the energy consumption of the TDMA scheme; i.e., there are no allocations that satisfy $(x+y) (P'_1+P'_2) \le xP_1+yP_2$.
\begin{figure}
    \centering
  \begin{subfigure}[b]{0.45\linewidth}
    \centering
    \includegraphics[width=0.8\linewidth]{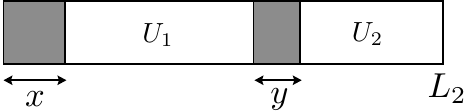} 
   \caption{TDMA scheme with two time slots}
    \label{TDMA_over_Ind_Dec_a} 
  \end{subfigure}
  \begin{subfigure}[b]{0.45\linewidth}
    \centering
    \includegraphics[width=0.8\linewidth]{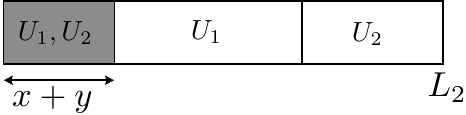} 
	\caption{Independent decoding scheme with three time slots}
    \label{TDMA_over_Ind_Dec_b} 
  \end{subfigure} 
  \caption{The structure of a TDMA scheme (a), and the corresponding independent decoding scheme (b).}
  \label{Time_Slots_Three_Two}
\end{figure}

As shown in Section~\ref{Ind_Dec}, the optimal transmission rates and transmission powers in the independent decoding scheme can be written as
\begin{gather}
\label{TDMA_over_Ind_Rates}
R'_1 = \log_2 \bigl( 1+\tfrac{\alpha_1 P'_1}{1+\alpha_2 P'_2}\bigr) \qquad R'_2 = \log_2 \bigl( 1+\tfrac{\alpha_2 P'_2}{1+\alpha_1 P'_1}\bigr).
\end{gather}
Considering \eqref{TDMA_over_Ind_Energy} and \eqref{TDMA_over_Ind_Rates} the problem of finding the rates for the independent decoding scheme such that the total energy consumption is less than that of TDMA can be written as
\begin{subequations}
\label{TDMA_over_INd_Dec_Opt_Prob}
\begin{align}
\text{find} \quad & P'_1, P'_2 \quad \\
\label{TDMA_over_INd_Dec_Opt_Prob_b}
\text{s.t.} \quad & \log_2 \bigl( 1+\tfrac{\alpha_1 P'_1}{1+\alpha_2 P'_2}\bigr) = \lambda \log_2(1+\alpha_1P_1), \\
\label{TDMA_over_INd_Dec_Opt_Prob_c}
\quad & \log_2 \bigl( 1+\tfrac{\alpha_2 P'_2}{1+\alpha_1 P'_1}\bigr) = (1-\lambda) \log_2(1+\alpha_2P_2), \\
\quad & P'_1+ P'_2 \le \lambda P_1+ (1-\lambda) P_2,\\
\quad & 0 \le \lambda \le 1,
\end{align}
\end{subequations}
where $\lambda = \tfrac{x}{x+y}$. Using the constraints in \eqref{TDMA_over_INd_Dec_Opt_Prob_b} and \eqref{TDMA_over_INd_Dec_Opt_Prob_c} we can rewrite $P'_1$ and $P'_2$ in terms of $P_1$ and $P_2$, namely, 
\begin{gather}
\label{P'_1_P'_2}
\alpha_1 P'_1 = \tfrac{\phi_2^{(1-\lambda)} (\phi_1^{\lambda}-1) }{\phi_1^{\lambda}+\phi_2^{(1-\lambda)}-\phi_1^{\lambda}\phi_2^{(1-\lambda)}}, \quad \alpha_2 P'_2 = \tfrac{\phi_1^{\lambda} (\phi_2^{(1-\lambda)}-1) }{\phi_1^{\lambda}+\phi_2^{(1-\lambda)}-\phi_1^{\lambda}\phi_2^{(1-\lambda)}}
\end{gather}
in which $\phi_1 = (1+\alpha_1P_1)$ and $\phi_2 = (1+\alpha_2P_2)$. Accordingly, the problem in \eqref{TDMA_over_INd_Dec_Opt_Prob} can be written as
\begin{subequations}
\label{TDMA_over_INd_Dec_Opt_Prob_lambda}
\begin{align}
\text{find} \quad & \lambda \quad \\ \nonumber
\label{TDMA_over_INd_Dec_Opt_Prob_lambda_b}
\text{s.t.} \quad & \tfrac{\bigl(\phi_2^{(1-\lambda)} (\phi_1^{\lambda}-1)\bigr)/\alpha_1 + \bigl(\phi_1^{\lambda} (\phi_2^{(1-\lambda)}-1)\bigr)/\alpha_2}{\phi_1^{\lambda}+\phi_2^{(1-\lambda)}-\phi_1^{\lambda}\phi_2^{(1-\lambda)}} \\ 
& \quad \quad \quad \quad \quad   \quad \quad  \le  \lambda P_1+ (1-\lambda) P_2 ,\\ 
\quad & 0 \le \lambda \le 1.
\end{align}
\end{subequations}
It can be seen that the right hand side of the constraint in \eqref{TDMA_over_INd_Dec_Opt_Prob_lambda_b} is a linear function of $\lambda$, and that for the values $\lambda=0$ and $\lambda=1$ the constraint holds with equality. To show that for values of $\lambda$ that lie within the interval $(0,1)$ the constraint in \eqref{TDMA_over_INd_Dec_Opt_Prob_lambda_b} cannot be satisfied, we write the first derivative of the left hand side function of \eqref{TDMA_over_INd_Dec_Opt_Prob_lambda_b} as
\begin{align} \nonumber
\label{derivative}
\tfrac{d}{d\lambda} = & \tfrac{\phi_1^{\lambda} \phi_2^{(1-\lambda)}}{(\phi_1^{\lambda}+\phi_2^{(1-\lambda)}-\phi_1^{\lambda}\phi_2^{(1-\lambda)})^2} \\ 
& \times \bigl( (\ln \phi_1 + \ln \phi_2) (\tfrac{1}{\alpha_1}-\tfrac{1}{\alpha_2}) - \tfrac{\ln \phi_2 }{\alpha_1} \phi_1^{\lambda}+\tfrac{\ln \phi_1 }{\alpha_2} \phi_2^{(1-\lambda)}\bigr),
\end{align}
which has a positive value for $\lambda = 0$ and a negative value for $\lambda =1$. In addition, there is only one value of $\lambda$ for which the derivative is equal to zero. This is because the function $f(\lambda) = \tfrac{\ln \phi_2 }{\alpha_1} \phi_1^{\lambda}-\tfrac{\ln \phi_1 }{\alpha_2} \phi_2^{(1-\lambda)} $ is increasing in terms of $\lambda$ and accordingly it has only one intersection with the constant value $(\ln \phi_1 + \ln \phi_2) (\tfrac{1}{\alpha_1}-\tfrac{1}{\alpha_2})$. Based on these observations we can conclude that for any value of $\lambda \in (0,1)$ the function in the left hand side of \eqref{TDMA_over_INd_Dec_Opt_Prob_lambda_b} is greater than the function in the right hand side. Hence, the objective function of TDMA has a value that is no larger than that of independent decoding if TDMA is feasible. }

\rev{\section{Optimality of Suboptimal Methods for Equal Channel Gains}
\label{app_equal_channel_gain}
Consider a two-time-slot system that employs the full multiple access scheme in the first time slot which is of the length $\tau_1 = \tilde{L}_1$, and a three-time-slot system that employs independent decoding in the first time slot. For the latter system the first user has the opportunity to use the first and the second time slots to transmit its task to the access point, and hence $\tau_1+\tau_2 = \tilde{L}_1$. It will suffice to assume that the total energy consumption and the number of transmitted bits in the last slot of both systems are equal, and to show that we can find a set of rate and power allocations and the slot durations for the independent decoding scheme so that in its first two time slots it transmits the required bits with the same energy as the first slot of the optimal scheme. Our resource allocation for the independent decoding scheme will adopt a TDMA structure, and hence the result applies to TDMA, too. 

Let $R_{k}$ and $P_{k}$ denote the parameters of the $k$\textsuperscript{th} user in the first slot of the full multiple access scheme and let $R'_{ki}$ and $P'_{ki}$ denote the corresponding parameters in the $i$\textsuperscript{th} time slot of the independent decoding scheme. In order to guarantee that the energy consumption and the transmitted bits in both schemes are equal, the independent decoding scheme must satisfy
\begin{subequations}
\label{equal_channel_optimality}
\begin{align}
\label{equal_channel_optimality_a}
& \tilde{L}_1 R_{1}  = \tau_1 R'_{11} + (\tilde{L}_1-\tau_1) R'_{12},\\
\label{equal_channel_optimality_b}
& \tilde{L}_1 R_{2}  = \tau_1 R'_{21},\\
\label{equal_channel_optimality_c}
& \tilde{L}_1 (P_{1}+P_2)  = \tau_1 (P'_{11}+ P'_{21}) + (\tilde{L}_1-\tau_1) P'_{12}.
\end{align}
\end{subequations}
Using the closed-form solutions obtained in Sections~\ref{general_case} and \ref{Ind_Dec}, we can rewrite \eqref{equal_channel_optimality_c} for equal channel gains as
\begin{equation}
\label{energy_ind_dec_opt}
\tfrac{\tilde{L}_1}{\alpha} 2^{R_{1}+R_2}  = \tfrac{\tau_1}{\alpha} \bigl( \tfrac{2^{R'_{11}+ R'_{21}}}{2^{R'_{11}}+ 2^{R'_{21}}-2^{R'_{11}+ R'_{21}}} \bigr) + \tfrac{\tilde{L}_1-\tau_1}{\alpha} 2^{R'_{12}},
\end{equation}
where $\alpha = \alpha_1 = \alpha_2$. It will suffice to set $R'_{11} = 0$, so the independent decoding scheme adopts a TDMA structure. In that case, \eqref{energy_ind_dec_opt} can be written as
\begin{equation}
\label{energy_ind_dec_opt_omit_R11}
\tfrac{\tilde{L}_1}{\alpha} 2^{R_{1}+R_2}  = \tfrac{\tau_1}{\alpha} 2^{R'_{21}}  + \tfrac{\tilde{L}_1-\tau_1}{\alpha} 2^{R'_{12}}.
\end{equation}
From \eqref{equal_channel_optimality_a} and \eqref{equal_channel_optimality_b}, we have that 
$R_{1}+R_{2}  = \tfrac{\tau_1}{\tilde{L}_1} R'_{21} + \tfrac{\tilde{L}_1-\tau_1}{\tilde{L}_1} R'_{12}$. Accordingly, \eqref{energy_ind_dec_opt_omit_R11} takes the form
\begin{equation}
\label{energy_ind_dec_opt_convex_style}
2^{(\lambda R'_{21} + (1-\lambda) R'_{12} )} =  \lambda 2^{R'_{21}}+(1-\lambda) 2^{R'_{12}},
\end{equation}
where $\lambda = \tau_1 / \tilde{L}_1$. Since the exponential function is strictly convex, if $\tau_1 \in (0, \tilde{L}_1)$ the equality in \eqref{energy_ind_dec_opt_convex_style} holds if and only if $R'_{12} = R'_{21}$. In that case, using \eqref{energy_ind_dec_opt_omit_R11}, $R'_{12} = R'_{21} = R_1 +R_2$, and from \eqref{equal_channel_optimality_b} $\textstyle{\tau_1 = \tilde{L}_1R_2 / (R_1+R_2)}$. If the power constraints satisfy $\textstyle{\bar{P}_1, \bar{P}_2 \ge (2^{R_1+R_2}-1)/{\alpha}}$, the rates $R'_{12}$ and $R'_{21}$ are achievable and TDMA (and hence independent decoding) can achieve the minimum energy consumption.}

\section{Quasi-convexity of the objective function in \eqref{prob_slots_partial_powers}}
\label{App_F_Quasi_convexity}
A function $f$ with a scalar argument is quasi-convex if at least one of the following conditions holds \cite{boyd2004convex}, 
\begin{itemize}
\item[(a)] $f$ is non-increasing,
\item[(b)] $f$ is non-decreasing,
\item[(c)] there is a (turning) point, $c$, such that for any $x \le c$ the function $f(x)$ is non-increasing and for any $x \ge c$ the function $f(x)$ is non-decreasing. 
\end{itemize}
%
Here, we will show that 
when $R_{21}$ and $R_{22}$ are held constant and the objective function in \eqref{prob_slots_partial_powers} is viewed as
a function of $R_{11}$, 
either the condition (b) or the condition (c) holds.
An analogous approach can be employed to prove the quasi-convexity with respect to $R_{21}$ and $R_{22}$.


For a given pair $(R_{21}, R_{22})$, the derivative of the objective with respect to $R_{11}$ can be written as
\begin{equation}
\label{factoring}
\tfrac{T_s \bar{L}_1}{(T_s + \delta_c R_{11})^2} \times F_r, 
\end{equation}
where
\begin{subequations}
\label{partial_derivative_to_R11_factorized}
\begin{align} \notag 
 F_r  =~ & \displaystyle{- \delta_c } \big(\tfrac{2^{R_{11}}-1}{\alpha_1}+ \tfrac{2^{R_{11}} (2^{R_{21}}-1)}{\alpha_2}\big)\\ \notag
& + \displaystyle{ (T_s+\delta_c R_{11}) } \big( \tfrac{2^{R_{11}} \operatorname{ln}2}{\alpha_1} + \tfrac{2^{R_{11}} \operatorname{ln}2 (2^{R_{21}}-1)}{\alpha_2}\big)\\ \notag
&  +\displaystyle{\big(\tfrac{\delta_c (T_s+\delta_c R_{21})}{T_s+\delta_c R_{23}}\big) \tfrac{2^{R_{23}}-1}{\alpha_2}} \\ \notag
& - \displaystyle{\mathcal{F}_1^{'}\big( B_1-\tfrac{\bar{L}_1 R_{11}}{T_s+\delta_c R_{11}} \big)} \\ \notag 
& + \displaystyle{\mathcal{F}_2^{'}\big( B_2- \tfrac{\bar{L}_2 R_{23}}{T_s+\delta_c R_{23}}-\tfrac{\bar{L}_1 T_s (R_{21}-R_{23})}{(T_s+\delta_c R_{11})(T_s+\delta_c R_{23})} \big)} \\  \notag
&\quad \times  \big( \tfrac{ \delta_c (R_{21}-R_{23})}{T_s+\delta_c R_{23}}\big).
\end{align}
\end{subequations}
As $\tfrac{T_s \bar{L}_1}{(T_s + \delta_c R_{11})^2}$ is always positive, 
to show that either condition (b) or condition (c) holds, it is sufficient to show that $F_r$ is non-decreasing.
In order to show that,
we will show that the derivative of $F_r$ with respect to $R_{11}$ is always non-negative. The derivative is  
\begin{subequations}
\label{partial_derivative_to_R11_sec}
\begin{align} \notag 
 \tfrac{dF_r}{dR_{11}} = ~&  \displaystyle{ (T_s+\delta_c R_{11}) } \big( \tfrac{2^{R_{11}} (\operatorname{ln}2)^2}{\alpha_1} + \tfrac{2^{R_{11}} (\operatorname{ln}2)^2 (2^{R_{21}}-1)}{\alpha_2}\big)\\ \notag
& + \displaystyle{\mathcal{F}^{''}_1\big( B_1-\tfrac{\bar{L}_1 R_{11}}{T_s+\delta_c R_{11}} \big)} 
\big( \tfrac{\bar{L}_1 T_s}{(T_s + \delta_c R_{11})^2}\big) \\ \notag 
& + \displaystyle{\mathcal{F}^{''}_2\big( B_2- \tfrac{\bar{L}_2 R_{23}}{T_s+\delta_c R_{23}}-\tfrac{\bar{L}_1 T_s (R_{21}-R_{23})}{(T_s+\delta_c R_{11})(T_s+\delta_c R_{23})} \big)} \\  \notag
& \quad \times  \big( \tfrac{ \delta_c (R_{21}-R_{23})}{T_s+\delta_c R_{23}}\big)^2 \big( \tfrac{\bar{L}_1 T_s}{(T_s + \delta_c R_{11})^2}\big).
\end{align}
\end{subequations}
Since $\mathcal{F}_k(\cdot)$ in \eqref{loc_opt_energy} is a convex function for non-negative arguments, its second derivative is non-negative.
Accordingly, $\tfrac{dF_r}{dR_{11}}$ is non-negative, and hence $F_r$ is non-decreasing. 
 
\changen{\section{Single-user Partial offloading} 
\label{single_user_partial}
The energy minimization problem for a system that employs the full multiple access scheme and has a single user seeking to complete its divisible task within a specific deadline takes a form that is similar to that in \eqref{prob_slots_partial_main} with the transmission rate, transmission power, and the portion of offloaded bits of that user as the only variables. In particular, the problem in single-user partial offloading can be written as  
 \begin{subequations}
\label{resp_1}
\begin{align} 
\min_{\substack {P_{1}, R_{1}, \gamma_{1}}} \quad &  \displaystyle{ (\tfrac{\gamma_{1} B_1}{R_{1}}) P_{1}} +\tfrac{M_1}{{L}_1^2} \bigl((1-\gamma_{1}) B_1\bigr)^3\\ 
\label{resp_1_b} 
\text{s.t.} \quad \quad & 0 \le \gamma_{1} \le 1,\\
 \label{resp_1_c}
\quad  & T_s (\tfrac{\gamma_{1} B_1}{R_{1}}) + \delta_c \gamma_{1}B_1 \le \bar{L}_1,\\
\label{resp_1_d}
& 0 \le P_{1} \le \bar{P}_1,\\
\label{resp_1_e}
\quad & 0 \le R_{1} \le \log_2 (1+ \alpha_1 P_{1}).
\end{align}
\end{subequations}
At optimality, the constraint in \eqref{resp_1_c} holds with equality,
\begin{equation}
\gamma_{1} = \tfrac{\bar{L}_1 R_{1}}{B_1 (T_s+\delta_c R_{1})},
\end{equation}
and for a given $R_{1}$ the closed-form optimal solution for transmission power is $P_{1} = \tfrac{2^{R_{1}}-1}{\alpha_1}$. Accordingly, the energy minimization problem in this case is reduced to the following single-variable optimization problem
\begin{subequations}
\label{prob_slots_single_user}
\begin{align}
\label{prob_slots_single_user_a}
\min_{\substack {R_{1}}} \quad &  \displaystyle{ \tfrac{ \bar{L}_1 }{\alpha_1}\bigl( \tfrac{2^{R_{1}}-1}{T_s+\delta_c R_{1}}}\bigr)+\tfrac{M_1}{{L}_1^2}  \big( B_1-\tfrac{\bar{L}_1 R_{1}}{T_s+\delta_c R_{1}} \big)^3\\ 
\label{prob_slots_single_user_b} 
\text{s.t.}    \quad & 0 \le \tfrac{\bar{L}_1 R_{1}}{B_1 (T_s+\delta_c R_{1})} \le 1,\\
\label{prob_slots_single_user_c}
& 0 \le R_{1} \le \log_2(1+\alpha_1 \bar{P}_{1}).
\end{align}
\end{subequations}
This is a convex optimization problem and can be efficiently solved. Indeed, the solution is either the point at which the derivative of the objective in \eqref{prob_slots_single_user_a} is zero (if that point is feasible) or one of the end points of the feasibility interval for $R_1$ imposed by \eqref{prob_slots_single_user_b} and \eqref{prob_slots_single_user_c}. If there is no such interval, the problem is infeasible.}
\changen{\section{Two-dimensional Mixed Offloading Problem}
\label{mixed_off_both}
The mixed offloading case, in which one user has an indivisible task and the other user has a divisible task, can be considered as a special case of the partial offloading problem with the offloading fraction of the task of the user with the indivisible task being set to one. However, the problem decomposition and the closed-form optimal solutions we have obtained in the partial offloading case enable us to simplify the energy minimization problem in the mixed offloading case by reducing dimension of the final optimization problem, problem \eqref{prob_slots_partial_powers}, by one. 

Let us first consider the case in which the user with the indivisible computational task has the shorter latency; i.e., user 1 is the user that is completely offloading its indivisible task. In that case, since user 1 offloads its complete task to the access point  $\gamma_{11} = 1$, and the energy minimization problem can be written as 

\begin{subequations}
\label{prob_app_mixed}
\begin{align}
\min_{\substack {P_{11}, P_{21}, P_{23}, \\ R_{11}, R_{21}, R_{23}, \\  \gamma_{21}, \gamma_{23}}} \quad &   \displaystyle{ (\tfrac{B_1}{R_{11}}) P_{11}} + \displaystyle{ (\tfrac{\gamma_{21} B_2}{R_{21}}) P_{21}}+\displaystyle{(\tfrac{\gamma_{23} B_2}{R_{23}}) P_{23}} + \displaystyle{\tfrac{M_2}{L_2^2}\big( (1-\gamma_{21} - \gamma_{23}) B_2\big)^3}\\ 
\label{prob_app_mixed_b} 
\text{s.t.}  \enspace \quad \quad & 0 \le \gamma_{21} \le 1,\quad  0 \le \gamma_{23} \le 1,\\
\label{prob_app_mixed_c}
 \quad &0 \le \gamma_{21}+\gamma_{23} \le 1,\\
 \label{prob_app_mixed_d}
\quad  & T_s (\tfrac{B_1}{R_{11}}) + \delta_c B_1 \le \bar{L}_1,\\
 \label{prob_app_mixed_e}
\quad  & T_s (\tfrac{\gamma_{21}}{R_{21}}+\tfrac{\gamma_{23}}{R_{23}}) B_2+ \delta_c (\gamma_{21}+ \gamma_{23})B_2 \le \bar{L}_2,\\
\label{prob_app_mixed_f}
& 0 \le P_{k1}, P_{k2} \le \bar{P}_k, \quad k=1,2,\\
\label{prob_app_mixed_g}
\quad & 0 \le R_{23} \le \log_2 (1+ \alpha_2 P_{23}),\\
\label{prob_app_mixed_h}
\quad & \{R_{11}, R_{21}\} \in \mathcal{R}.
\end{align}
\end{subequations}
As was explained in Section~\ref{partial_mac}, at optimality the constraints in \eqref{prob_app_mixed_d} and \eqref{prob_app_mixed_e} hold with equality and the closed-form optimal solution for the transmission rate of the user with indivisible task is obtained as $R_{11} = \tfrac{T_s B_1}{\bar{L}_1-\delta_c B_1}$. Following the same steps as those in Section~\ref{partial_mac}, the closed-form optimal solutions for $\gamma_{21}$ and $\gamma_{23}$ can be obtained. Therefore, for the case in which the channel gains are such that $\tfrac{\alpha_1}{\alpha_2} \le 1$, the final optimization problem, which is analogous to \eqref{prob_slots_partial_powers}, can be written as
\begin{subequations}
\label{prob_slots_partial_powers_mixed}
\begin{align} \notag
\min_{R_{21}, R_{23}} \; &    \displaystyle{ \bigl( \tfrac{\bar{L}_1-\delta_c B_1}{T_s} \bigr)} \bigl(\tfrac{2^{\frac{T_s B_1}{\bar{L}_1-\delta_c B_1}}-1}{\alpha_1}+ \tfrac{2^{\frac{T_s B_1}{\bar{L}_1-\delta_c B_1}} (2^{R_{21}}-1)}{\alpha_2}\bigr)  \\
& \quad \quad  +\tfrac{1}{T_s+\delta_c R_{23}}\displaystyle{\big(\bar{L}_2 - \tfrac{(T_s+\delta_c R_{21})}{T_s}(\bar{L}_1-\delta_c B_1) \big) \bigl(\tfrac{2^{R_{23}}-1}{\alpha_2}\bigr)} \\ 
& \quad \quad \quad \quad  +\tfrac{M_2}{{L}_2^2} \Big( B_2- \tfrac{\bar{L}_2 R_{23} + (R_{21}-R_{23}) (\bar{L}_1 - \delta_c B_1)}{T_s+\delta_c R_{23}} \Big)^3
\label{prob_slots_partial_powers_mixed_a} 
\\ \label{prob_slots_partial_powers_mixed_b} 
\text{s.t.}  \enspace \quad \quad   & \eqref{prob_app_mixed_b}, \eqref{prob_app_mixed_c},\\
\label{prob_slots_partial_powers_mixed_c}
\quad & 0 \le R_{2i} \le \log_2 (1+ \alpha_2 \bar{P}_{2}), \quad i= 1, 3.
\end{align}
\end{subequations}
This problem has a similar structure to the problem in \eqref{prob_slots_partial_powers}; i.e., the objective function is quasi-convex in each variable when the other variable is given. Therefore, a coordinate descent algorithm can be employed to obtain a stationary solution. However, in this ``mixed'' case we have been able to obtain a closed-form expression for $R_{11}$ and hence the problem in \eqref{prob_slots_partial_powers_mixed} has only two variables, whereas the problem in \eqref{prob_slots_partial_powers} has three. 

For the case in which the latency of the user with the indivisible task is longer than that of the user with the divisible task, the energy minimization problem is similar to \eqref{prob_slots_partial_powers}, but has the constraint \eqref{prob_slots_partial_main_d} changed to $\gamma_{21}+\gamma_{23} = 1$ rather than $0 \le \gamma_{21}+\gamma_{23} \le1$. This equality constraint enables us to reduce dimension of the final optimization problem by one. To show that, we observe that since the latency constraint on the user with the divisible task (user 1) will hold with equality at optimality and since the equality $\gamma_{21}+\gamma_{23} = 1$ holds, the closed-form solutions for $\gamma_{11}$, $\gamma_{21}$, and $\gamma_{23}$ are
 \begin{subequations}
\label{gamm_equs_app}
\begin{align}
\label{gamm_equs_app_11}
&\gamma_{11} = \tfrac{\bar{L}_1 R_{11}}{B_1 (T_s+\delta_c R_{11})},\\
\label{gamm_equs_app_21}
&\gamma_{21} = \tfrac{\bar{L}_1 R_{21}}{B_2 (T_s+\delta_c R_{11})},\\
\label{gamm_equs_app_23}
&\gamma_{23} = 1- \tfrac{\bar{L}_1 R_{21}}{B_2 (T_s+\delta_c R_{11})}.
\end{align}
\end{subequations}
By using the equality of the latency constraint on the user with the indivisible task (user 2) at optimality, we can find the closed-form optimal solution for either $R_{21}$ or $R_{23}$. By substituting the closed-form optimal solution for $R_{23}$,
\begin{equation}
R_{23} = \tfrac{T_s B_2(T_s+\delta_cR_{11})-T_s \bar{L}_1R_{21}}{(T_s+\delta_c R_{11})(\bar{L}_2-\delta_cB_2)- T_s\bar{L}_1},
\end{equation}
and the optimal values for $\gamma_{ki}$ in \eqref{gamm_equs_app} into the problem in \eqref{prob_slots_partial_powers}, the remaining optimization problem can be written as
\begin{subequations}
\label{prob_slots_mixed_powers}
\begin{align} \notag
\min_{R_{11}, R_{21}} \; &    \displaystyle{ \tfrac{ \bar{L}_1 }{T_s+\delta_c R_{11}} } (\tfrac{2^{R_{11}}-1}{\alpha_1}+ \tfrac{2^{R_{11}} (2^{R_{21}}-1)}{\alpha_2})  \\ \notag
& \quad  +\displaystyle{\tfrac{1}{\alpha_2} \big(\tfrac{\bar{L}_2-\delta_cB_2}{T_s} - \tfrac{ \bar{L}_1}{T_s + \delta_c R_{11} }\big) {\big(2^{\tfrac{T_s B_2(T_s+\delta_cR_{11})- T_s \bar{L}_1R_{21}}{(T_s+\delta_c R_{11})(\bar{L}_2-\delta_cB_2)- T_s\bar{L}_1}}-1}\big)} \\ 
& \quad \quad  +\tfrac{M_1}{{L}_1^2}  \big( B_1-\tfrac{\bar{L}_1 R_{11}}{T_s+\delta_c R_{11}} \big)^3\\
 \label{prob_slots_mixed_powers_b} 
\text{s.t.}  \enspace \quad \quad   & \eqref{prob_app_mixed_b}, \eqref{prob_app_mixed_c},\\
\label{prob_slots_mixed_powers_c} 
\quad &0 \le \tfrac{T_s B_2(T_s+\delta_cR_{11})-T_s \bar{L}_1R_{21}}{(T_s+\delta_c R_{11})(\bar{L}_2-\delta_cB_2)- T_s\bar{L}_1} \le \log_2 (1+ \alpha_2 \bar{P}_{2}),\\
\label{prob_slots_mixed_powers_d}
\quad & 0 \le R_{11} \le \log_2 (1+ \alpha_1 \bar{P}_{1}),\\
\label{prob_slots_mixed_powers_c}
\quad & 0 \le R_{21} \le \log_2 (1+ \tfrac{\alpha_2 \bar{P}_{2}}{2^{R_{11}}}).
\end{align}
\end{subequations}
This problem has similar structure, with one dimension reduced, to the problem in \eqref{prob_slots_partial_powers}. In particular, the objective function is quasi-convex in each variable when the other variable is given. Hence, a coordinate descent algorithm can be employed to obtain a stationary solution. }

\bibliographystyle{IEEEtran}
\bibliography{references}

\end{document}